\documentclass[aoas]{imsart}

\RequirePackage{amsthm,amsmath,amsfonts,amssymb}
\RequirePackage[authoryear]{natbib}
\RequirePackage[colorlinks,citecolor=blue,urlcolor=blue]{hyperref}
\RequirePackage{graphicx}
\RequirePackage{xspace}
\RequirePackage[colorlinks,citecolor=blue,urlcolor=blue]
{hyperref}

\RequirePackage{booktabs,array,multirow}
\RequirePackage{algorithm,algorithmic}
\usepackage[most]{tcolorbox}

\startlocaldefs
\theoremstyle{plain}

\newtheorem{theorem}{Theorem}[section]

\newtheorem{proposition}[theorem]{Proposition}

\theoremstyle{remark}
\newtheorem{definition}[theorem]{Definition}

\def\argmax{\mathop{\mathrm{arg\,max}}} 
\newcommand*\diff{\mathop{}\!\mathrm{d}}

\def\XS{\xspace}
\DeclareMathAlphabet{\mathb}{OML}{cmm}{b}{it}
\def\sbm#1{\ensuremath{\mathb{#1}}}                
\def\sbmm#1{\ensuremath{\boldsymbol{#1}}}          

\def\Nset{\mathbb{Z}_+} 
\def\Rset{\mathbb{R}} 

\def\Ab{{\sbm{A}}\XS}  \def\ab{{\sbm{a}}\XS}

\def\Db{{\sbm{D}}\XS}  \def\db{{\sbm{d}}\XS}
\def\Eb{{\sbm{E}}\XS}

\def\Ib{{\sbm{I}}\XS}  
\def\Jb{{\sbm{J}}\XS}

\def\Qb{{\sbm{Q}}\XS}  
  \def\rb{{\sbm{r}}\XS}
\def\Sb{{\sbm{S}}\XS}  \def\sb{{\sbm{s}}\XS}
  
\def\Ub{{\sbm{U}}\XS}  
\def\Vb{{\sbm{V}}\XS}

\def\thetab      {{\sbmm{\theta}}\XS}      \def\Thetab    {{\sbmm{\Theta}}\XS}

\def\mub         {{\sbmm{\mu}}\XS}

      \def\Sigmab    {{\sbmm{\Sigma}}\XS}
  
\def\phib        {{\sbmm{\phi}}\XS}

\newcommand{\tv}{\ensuremath{\mathbf{t}}} 
\newcommand{\sv}{\ensuremath{\mathbf{s}}} 
\newcommand{\Wv}{\ensuremath{\mathbf{W}}} 
\newcommand{\Xv}{\ensuremath{\mathbf{X}}} 
\newcommand{\xv}{\ensuremath{\mathbf{x}}} 
\newcommand{\Yv}{\ensuremath{\mathbf{Y}}} 
\newcommand{\Zv}{\ensuremath{\mathbf{Z}}} 
\newcommand{\yv}{\ensuremath{\mathbf{y}}} 
\newcommand{\zv}{\ensuremath{\mathbf{z}}} 

\newcommand{\tun}{T\textsubscript{1}} 
\newcommand{\tdeux}{T\textsubscript{2}}

\endlocaldefs

\begin{document}

\begin{frontmatter}
\title{Scalable magnetic resonance fingerprinting: Incremental inference of high-dimensional elliptical mixtures  from  large data volumes}
\runtitle{Scalable magnetic resonance fingerprinting}

\begin{aug}
\author[A,B,C]{\fnms{Geoffroy}~\snm{Oudoumanessah}\ead[label=e1]{geoffroy.oudoumanessah@inria.fr}\orcid{0000-0002-2240-1169}},
\author[B]{\fnms{Thomas}~\snm{Coudert}\ead[label=e2]{thomas.coudert@inserm.fr}\orcid{0009-0006-0158-6105}}
\author[C]{\fnms{Carole}~\snm{Lartizien}\ead[label=e3]{carole.lartizien@creatis.insa-lyon.fr}\orcid{0000-0001-7594-4231}}
\author[A,B]{\fnms{Michel}~\snm{Dojat}\ead[label=e4]{michel.dojat@inserm.fr}\orcid{0000-0003-2747-6845}}
\author[B]{\fnms{Thomas}~\snm{Christen}\ead[label=e5]{thomas.christen@univ-grenoble-alpes.fr}\orcid{0000-0002-0498-9296}}
\and
\author[A]{\fnms{Florence}~\snm{Forbes}\ead[label=e6]{florence.forbes@inria.fr}\orcid{0000-0003-3639-0226}}

\address[A]{Univ. Grenoble Alpes, Inria, CNRS, Grenoble INP, LJK, 38000 Grenoble, France \printead[presep={,\ }]{e1,e6}}

\address[C]{Univ. Lyon, CNRS, Inserm, INSA Lyon, UCBL, CREATIS, UMR5220, U1294, F‐69621, Villeurbanne, France\printead[presep={,\ }]{e3}}

\address[B]{Univ. Grenoble Alpes, Inserm U1216, CHU Grenoble Alpes, Grenoble Institut des Neurosciences, 38000 Grenoble, France\printead[presep={,\ }]{e2,e4,e5}}
\end{aug}

\begin{abstract}
Magnetic Resonance Fingerprinting (MRF) is an emerging technology with the potential to revolutionize radiology and medical diagnostics. In comparison to  traditional magnetic resonance imaging (MRI), MRF enables the rapid, simultaneous, non-invasive acquisition and reconstruction of multiple  tissue parameters, paving the way for novel diagnostic techniques. 
In the original {\it matching} approach, reconstruction is based on the search for the best matches between \textit{in vivo} acquired  signals
and a dictionary of high-dimensional simulated signals  (fingerprints) with known tissue properties.
A critical and limiting challenge is that the size of the simulated dictionary increases exponentially with the number of parameters, leading to an extremely costly matching.  
In this work, we propose to address this scalability issue by considering probabilistic mixtures of high-dimensional elliptical distributions, to learn more efficient dictionary representations. Mixture components are modelled  as flexible elliptical shapes in low dimensional subspaces. They are exploited to cluster similar signals and reduce their dimension locally cluster-wise limiting information loss.  To estimate such a mixture model, we  provide a new incremental algorithm capable of handling large numbers of signals, allowing us to go far beyond the hardware limitations encountered by standard implementations.
We demonstrate, on simulated and real data, that our method  effectively manages large volumes of MRF data with maintained accuracy.  It offers a more efficient solution for accurate tissue characterization and  significantly reduces the computational burden, making the clinical application of MRF more practical and accessible.
\end{abstract}

 
\end{frontmatter}

\section{Introduction}
Traditional Magnetic Resonance (MR) imaging relies on an analytical resolution of dynamical equations using conventional tuning of the MR hardware through sequences of pulses, each characterized by different values of parameters such as the flip angle and repetition time.  Standard quantitative MRI (qMRI) methods are based on a single sequence for a single parameter measurement at a time. This leads to high scan times for multi-parametric protocols as each parameter estimate involves one MR sequence. A recent approach named Magnetic Resonance Fingerprinting (MRF, \cite{ma2013}) has been developed to overcome these limitations. The MRF protocol involves fast undersampled acquisitions with time-varying parameters defining the MRF sequence that produces temporal signal evolutions (named {\it fingerprints}) in each voxel. In the original proposal, a dictionary search approach is used to compare the \textit{in vivo} fingerprints with millions of numerical simulations of MR signals for which the associated parameters are known. These millions of simulated signals compose the so-called dictionary. The values of the parameters corresponding to the closest simulated signals or {\it matches} are then assigned to the associated {\it in vivo} voxels, allowing the simultaneous reconstruction of multiple quantitative maps (images) from extremely undersampled raw images, using only one single sequence, thus saving considerable acquisition time \citep{poorman2020,mcgivney2020}. From this, standard relaxometry MRF allows reconstructing parameter maps for relaxation times T\textsubscript{1} and T\textsubscript{2}, over the whole human brain (1 mm$^3$ spatial resolution) in 3 min \citep{ye2017, gu2018} compared to 30 min for a standard T1/T2 exam. Moreover, the flexibility of the numerical simulations enables correction of system imperfections as well as some patient motions by including them in the model and post-processing pipelines \citep{bipin2019}. Thus,  MRF could be a game changer for emergency  patients who need to complete exams in a few minutes.  The power of the MRF approach is not limited to the estimation of relaxation times, in theory, it allows the measurement of any parameter that influences nuclear magnetization ({\it e.g.}, microvascular networks), and could be added to the simulation model \citep{Wang2019,coudert2024,coudert2025mr}. However, increasing the number of estimated parameters, even moderately,  induces the design of more complex sequences and increased reconstruction times, from hours to days. This limits the clinical application of high-dimensional MRF and necessitates the development of innovative processing methods. 
Consequently, a significant focus is on improving MRF reconstruction methods, as reviewed by \cite{tippareddy2021} and \cite{monga2024}.
In this work, we propose to focus on reducing the reconstruction times of MRF when more than the main two parameters, $T_1$, $T_2$, are involved, including the addition of $\delta \! f$ the frequency offset, the sensitivity of the magnetic field $B_1$, cerebral blood volume ($\text{CBV}$), and microvascular geometry (e.g., vessel radius denoted as $R$).

MRF reconstruction is first recast as an inverse problem that can be solved using different approaches as recalled in the next section.
All approaches make use, at some stage, of a dictionary of simulated pairs (parameters, signal), which represent our knowledge of the link between tissue parameters and MR time series, through a so-called {\it direct} or {\it forward} model. The dictionary is then either used to learn an {\it inversion} operator, from signal to tissue parameters,  or to search for the best fits between observed signals and simulated ones. The approaches' scalability relies thus greatly on their ability to extract efficiently the information encoded in simulations. Efficiency has different aspects: for search-type methods, the dictionary should not be too big, while for learning-based methods, the dictionary should be informative enough. These potentially opposite requirements call for efficient representations of simulated data.
We propose to explore a divide-and-conquer strategy by introducing the framework of High Dimensional Mixtures of Elliptical Distributions (HD-MED).
Probabilistic mixtures of high-dimensional elliptical distributions allow us to learn more efficient dictionary representations. Mixture components are modeled as flexible elliptical shapes in low-dimensional subspaces. They are exploited to cluster similar signals and reduce their dimension locally, at the cluster level, limiting information loss. To estimate such a mixture model, we provide a new incremental algorithm capable of handling large numbers of signals, allowing us to go far beyond the hardware limitations encountered with standard implementations.
We demonstrate, on simulated and real data, that our method effectively manages large volumes of MRF data with maintained accuracy. It offers a more efficient solution for accurate tissue characterization and significantly reduces computational burden, making the clinical application of MRF more practical and accessible.
  
In the rest of the paper, we first recall in Section \ref{sec:IP} the two main types of approaches that have been investigated in the literature on MRF reconstruction. We then review related work on efficient dictionary representations and specify our contributions in Section \ref{sec:RW}. The mixtures of elliptical distributions are presented in Section \ref{sec:Mix} with their use for dimension reduction in Section \ref{Sec:Red}. The incremental algorithm proposed for the mixture estimation is presented in Section \ref{sec:oem} and illustrated on MRF in Section \ref{sec:MRF}.

\section{MRF reconstruction as an inverse problem}
\label{sec:IP}
An inverse problem refers to a situation where one aims to determine the causes of a phenomenon based on experimental observations. 
In MRF, the goal is to infer a set of tissue characteristics or parameter values from an observed signal, identifying those that best explain the signal. Compared to standard MRI, which relies on sequential measurements with fixed parameters, MRF acquires data using a sequence of pseudo-random parameters ({\it e.g.}, varying repetition times or flip angles). This process generates tissue-dependent time-series signals, where each time-series $\yv$ - known as a {\it fingerprint} - is assigned to a specific 3D location (voxel) in the brain.
Such a resolution generally starts with modeling the phenomenon under consideration, which is called the "direct" or "forward problem."
It is generally assumed that at least the numerical evaluation of the forward model is available because experts have designed equations that can be solved either analytically or numerically.
The most common use of the forward model is via a simulator that allows the creation of a database ${\cal D}_f$, usually referred to as a dictionary, of $N$ signals ${\yv_1, \ldots, \yv_N}$ with $\yv_i \in {\cal Y} \subset \Rset^M$, generated (stored or computed on the fly) by running the theoretical (physical) model $f$ for many different tissue parameter values ${\tv_1, \ldots, \tv_N}$ with $\tv_i \in {\cal T} \subset \Rset^L$ and $\yv_i = f(\tv_i)$. The typical dimension of each $\yv_i$ considered in this paper is $M=260$.
The generated tissue parameter values then only partially represent the full space ${\cal T}$ of possible values and correspond to a discrete grid in the space. In this context, we can distinguish two types of methods, referred to below as optimization and learning approaches.

\subsection{Optimization vs learning (or regression) approaches}

Optimization approaches include the most used method in MRF, namely dictionary matching, often called, in other domains, grid search, look-up-table or k-nearest neighbors. 
They consist of minimizing over parameters  $\tv$ a
merit function  $d$ expressing the similarity between the observed signal $\yv_{obs}$ and simulated
signal $f(\tv)$,
\begin{eqnarray}  \hat{\tv} &\in &  \arg \min_{\tv \in {\cal T}} d(\yv_{obs}, f, \tv) \label{modelM} \; .
\end{eqnarray}
Typically, $d(\yv_{obs}, f, \tv) = d(\yv_{obs}, f(\tv))= ||\yv_{obs}- f(\tv)||^2\; .$
Solutions are searched in the full ${\cal T}$ space but solutions could also be penalized as done in grid search methods. In grid search, 
the previous full search is replaced by a simpler look-up or matching operation making use of  ${\cal D}_f$, often created beforehand. The search space is significantly reduced from  a continuous space ${\cal T}$ to a discrete and finite ${\cal D}_f$. 
 The speed gain is significant in comparison to traditional optimization methods as retrieving a value from memory is often faster than undergoing an expensive usually iterative computation. 
Their disadvantage  is the instability of solutions.
Many questions  remain on how to choose the merit function, how many $\yv_n$ in
the look-up table have to be kept to estimate parameters, how to choose the look-up
table, {\it etc. }
When the number of parameters is small, grid search is suitable and can provide very good predictions. However, for even moderate numbers of parameters, the required number of elements in the dictionary renders grid search either intractable or inaccurate. The technique is not amortized, for each new $\yv_{obs}$, we have to compute the matching score $d(\yv_{obs}, \yv_n)$ for all $\yv_n$ in the dictionary.  When the dimension of $\tv$ increases, the dictionary (and $N$) has to  be larger too for better accuracy. The computation of $N$ matching scores can become time-consuming. 

 Regression or learning methods are more efficient in that sense and  usually have better amortization properties.
 These methods have the advantage of  providing tractable solutions in the case of massive inversions of high-dimensional data.
The main principle is to transfer the computational cost and time from individual pointwise predictions to the learning of a global inverse operator from ${\cal D}_f$.
The advantage is that once the operator is learned, it can be used, at negligible cost, for very large numbers of new signals. Then, the nature of the inverse operator needs to be specified.
Traditional learning or regression methods are not specifically designed for high-dimensional data but there has been a large literature covering this case, see {\it e.g.} \cite{Giraud2014} for a review. 
In MRF, a popular approach from \cite{cohen2018}, uses a four-layer fully connected network (DRONE) to learn the dictionary signals for reconstructing  $T_1$ and $T_2$ parameters. Eventhough DRONE provides good results, this perceptron-based model loses the temporal coherence of the signal. Recently, \cite{cabini2024} proposed a Recurrent Neural Network (RNN) with long-short term memory (LSTM) blocks, which yields better reconstruction results for $T_1$ and $T_2$ with more robustness towards noisy acquisitions. However, when more parameters need to be estimated, such as vessel oxygenation or radius \citep{christen2013, christen2014,coudert2025mr}, the dimensionality of the signals dramatically increases. \cite{barrier2024} demonstrated that a simple RNN is prone to catastrophic forgetting, where the RNN well estimates the beginning of the signals but learns the end less effectively. To mitigate this issue, they proposed to use bidirectional LSTM (bi-LSTM) blocks within the RNN architecture, ensuring that both the beginning and the end of the signal are efficiently learned by the new bi-LSTM blocks. Despite the good results in reconstructing $T_1$, $T_2$, and vascular parameters, bi-LSTM still faces challenges in learning the  middle
part of the signals.

Generally, the dictionary is seen as a collection of simulated signals with no particular spatial correlation. Another way to simulate signals is to acquire real parametric maps from different subjects and then simulate the MRF images. This approach was first proposed by \cite{soyak2021} and later improved by \cite{gu2024}. Both studies utilize a UNet \citep{ronneberger2015} to infer $T_1$ and $T_2$ directly using the entire MRF image as an input to the network, preserving valuable spatial information. Given the high-dimensionality of the signals, the authors proposed adding attention layers \citep{vaswani2017} to focus on the most important dimensions of the signals. More recently, \cite{li2024} highlighted the limitations of using CNNs, which have a restricted receptive field and capture spatial information only locally. To overcome these limitations, the authors proposed using a Local-Global vision Transformer to capture spatial information globally as well. However, capturing spatial information has a cost. Indeed, one needs to acquire, from a large group of subjects, multiple $T_1$ and $T_2$ maps, which takes about 30 min for a complete exam, making the data acquisition process much more costly than using more conventional dictionaries. Additionally, since $T_1$ and $T_2$ need to be acquired at two different times, this method introduces more errors from the registration of the acquired maps.

Finally, \cite{boux2021} also proposed to use  GLLiM \citep{Deleforge2015}, a model that casts MRF reconstruction into a Bayesian inverse problem and then solves it using a learning approach. This method takes into account the high-dimensional property of MRF signals by defining the low-dimensional variables as the regressors, which in our case are the tissue parameters. By doing so, they start learning the \textit{low-to-high} regression model from which they can derive the forward model parameters and then the \textit{high-to-low} regression model,  from MRF signals to tissue parameters, as desired.

\section{Related work and positioning}
\label{sec:RW}

In practice, acquired MRF acquisitions come as a 4D matrix, made of a time series of 3D MRI images where each voxel contains the acquired, potentially long, fingerprint signal. In this work, we focus on designing efficient representations of large highly precise grids of simulated signals counterparts.

\subsection{Parsimonious representations of dictionaries}
The curse of dimensionality goes with what is often called the bless of dimensionality, which refers to the fact that in high-dimensional data sets, useful information actually lives in much smaller-dimensional parts of the data space. 
An approach to the MRF reconstruction problem is then to reduce the dimension of the dictionary beforehand to reduce the matching or learning cost. The first to propose such an approach for efficient matching were \cite{mcgivney2014}, who applied Singular Value Decomposition (SVD) to the dictionary of signals. Once the decomposition is learned, one can project any new  acquisition into the SVD low-dimensional subspace. However, when the size of the dictionary increases, computing the SVD becomes costly as it requires loading the complete dictionary into fast-access memory (\textit{e.g.} the RAM). To address this, \cite{yang2018} proposed using randomized SVD \citep{halko2011}. \cite{golbabaee2019} also suggested applying SVD to the dictionary before training a neural network. Other methods proposed a non-Euclidean analysis of the dictionary space projecting the signals into a lower-dimensional manifold \citep{li2023}.
Reducing the dimension of high-dimensional dictionaries assumes that most of the information in the signals can be captured and represented in a much lower-dimensional subspace. Classical techniques include principal component analysis (PCA, \cite{jolliffe2016}), probabilistic principal component analysis (PPCA, \cite{tipping1999ppca}), factor analyzers (FA) \citep{Freitas2004}, sparse models \citep{Zou2006,daspremont2007,archambeau2008a}, and newer methods such as diffusion maps \citep{coifman2006}. More flexible approaches are based on mixtures of the previous ones, such as mixtures of factor analyzers (MFA, \cite{mclachlan2003}) introduced by \cite{gharamanietal97,gharamanietal99} and extended by \cite{mclachlan2003,baek2011mixtures}, and mixtures of PPCA (MPPCA, \cite{tipping1999mppca,xu2023}) with recent generalizations \citep{hong2023,xu2023}. Another mixture approach is called HDDC in \cite{bouveyron2007} and HD-GMM in \cite{bouveyron2014} for High Dimensional Gaussian Mixture Models, which encompass many forms of MFA and MPPCA and generalize them. In particular, HD-GMM can be used to obtain multiple low-dimensional subspaces of different dimensions. For a review on high-dimensional clustering via mixtures, see \cite{bouveyron2014}.
Such a divide-and-conquer strategy has been used in MRF by first creating a global clustering of the signals and then reducing the dimension. This method was applied by \cite{cauley2015}, who performed K-way partitioning of the dictionary before using multiple cluster-associated PCA to reduce the dimension in each part. When matching a new signal, one first determines which cluster it belongs to and then applies the projection learned by the associated PCA model. More recently, \cite{ullah2023} proposed a simpler approach: applying a clustering algorithm to the dictionary, before utilizing GPUs for the \textit{dictionary matching}, enabling fast matching without the need for dimensionality reduction. However, this method still faces issues when the number of parameters to estimate is larger than $3$.

Most of these methods are designed for batch data and are thus sensitive to hardware limits such as memory, restricting the amount of data they can process. For instance, some dictionaries exceed terabyte (TB) sizes. A simple solution is to down-sample data sets before processing, potentially losing useful information. Another approach is to design incremental, also referred to as online, variants that handle data sequentially in smaller groups. 
Although such approaches can be used for genuine streaming data that is not stored after being seen, in this work we consider cases where past samples can possibly be revisited.
A number of incremental approaches exist for dimension reduction techniques, see the recent SHASTA-PCA \citep{gilman2023bis} and references therein, or \cite{balzano2018} for a review. To our knowledge, far fewer solutions exist for mixtures. Estimation of such models is generally based on maximum likelihood estimation via the Expectation-Maximization (EM) algorithm (\cite{mcLachlan08}). A preliminary attempt for an incremental MPPCA can be found in \cite{bellas2013} but it is based on heuristic approximations of the EM steps.

\subsection{Contribution}

Herein, we propose to explore this divide-and-conquer strategy by introducing the framework of High Dimensional Mixtures of Elliptical Distributions (HD-MED) to overcome the challenge of estimating more than $2$ parameters, {\it e.g.} $T_1$, $T_2$, $\delta \! f$, $B_1$, $\text{CBV}$, $R$.
Considering the reconstruction of 6 parameters makes the associated dictionaries both large in size (an order of $13$ terabytes of signals to represent the tissue heterogeneity) and dimension, framing our proposal around two key components. Building on HD-MED, a generalization of HD-GMM, we show how they can be used to simultaneously compress and cluster large-scale high-dimensional MRF dictionaries, and more generally any dataset. We derive a new online algorithm, based on a principled EM framework, to learn such a model from  very large data volumes. We demonstrate the effectiveness of our approach on MRF reconstruction, showing results comparable to the high-dimensional dictionary matching referred to as \textit{full-matching} in the next sections. This approach allows us to exceed the resolution, size used in current implementations and to reconstruct a larger number of MR parameter maps with improved accuracy, thereby advancing the clinical feasibility of MRF. 

\section{High Dimensional Mixtures of Elliptical Distributions}
\label{sec:Mix}
Due to time constraints and computational complexity, data may be acquired with errors and at high noise levels. In such cases, using Gaussian models, which are known to be sensitive to outliers, is not recommended.  Elliptical distributions represent a family of distributions that contains Gaussian distributions but also heavy-tailed distributions, such as the Student or Laplace distributions. 
Elliptical distributions allow more flexible modeling, with additional parameters and better fit to data than Gaussian distributions, see {\it e.g.} \citep{Lange1989,Lange1993,Neri2021,Kotz2001}. This additional flexibility is often referred to as robustness because the heavier tails make estimation less sensitive to outlying samples.
 In this paper, we consider a sub-class of elliptical distributions, which can be expressed as infinite mixtures of Gaussian distributions.

\subsection{Gaussian scale mixture (GSM) distributions}
Scale mixtures of multivariate Gaussian distributions, also referred to as normal/independent distributions \citep{Lange1993}, are an  important subclass of elliptical distributions whose definition is recalled below. Scale mixtures of Gaussians share good properties with Gaussian distributions. They are tractable, lead to tractable inference procedures and provide more robust results, in contrast to Gaussian distributions that usually suffer from sensitivity to outliers.  
\begin{definition}[Elliptical Distributions (ED)]
A continuous random vector $ \Yv \in \mathbb{R}^M$ follows a multivariate elliptical symmetric distribution if its probability density function (pdf) $p(\yv)$ is of the following form (see \citet{cambanis} or \citet{kelker}),
\begin{align}
    p(\yv)= C_{p,g} |{\Sigmab}|^{-1/2} g\left((\yv-\mub)^{\top}{\Sigmab}^{-1}(\yv-\mub)\right),  \label{eq:ed}
\end{align}
where $\Sigmab \in \mathbb{R}^{M \times M}$ is the scale matrix with determinant $|\Sigmab|$, $\mub \in \mathbb{R}^M$ is the location or mean vector, $C_{p,g}$ is a normalizing constant such that the pdf $p(\yv)$ integrates to one. The non-negative function $g$ is called the density generator and determines the shape of the pdf. When $\Yv$ has density (\ref{eq:ed}),  we write $\Yv \sim \mathcal{E}_M\left(\mub, \Sigmab, g\right)$.
\end{definition}
Note that the scale matrix $\Sigmab$ is not necessarily the covariance matrix, $\Sigmab$ is proportional to the covariance matrix if the latter exists. 
The pdf of the multivariate normal distribution is a special case of  ED with $g(u) = \exp(-u^2/2)$. 
Another member of the elliptical family is the multivariate 
Student distribution. This distribution is well-studied in the literature \citep{kotz2004multivariate} and admits a useful representation as a Gaussian scale mixture.
Denoting ${\cal N}_M(\yv; \mub,\Sigmab)$ a $M$-variate Gaussian distribution with mean $\mub$ and covariance matrix $\Sigmab$, 
a Gaussian scale mixture distribution is a distribution of the following form.

\begin{definition}[Gaussian scale mixture  distributions (GSM)]
\label{def:GSM}
If $\mub$ is an $M$-dimensional vector, $\Sigmab$ is an $M \times M$ positive definite symmetric matrix and $f$ is
a  pdf of a univariate positive variable $W\in \Rset^+$, then the $M$-dimensional density  given by
 \begin{align}
p(\yv) = \int_{\Rset^+} {\cal N}_M\left(\yv ;  \mub, \frac{\Sigmab}{w}\right)  \; f(w)  \; dw \label{eq:GSM}
\end{align}
is said to be an infinite mixture of scaled
Gaussians or Gaussian scale mixture (GSM) with mixing distribution function $f$. If vector $\Yv$ has density (\ref{eq:GSM}), we  still 
write $\Yv \sim {\cal E}_M\left(\mub, \Sigmab,f \right)$ and refer to $W$ as the mixing variable.
\end{definition}
In practice, we will consider mixing distribution $f_\thetab$ that depends on some parameter $\thetab$ and also write $\Yv \sim {\cal E}_M\left(\mub, \Sigmab,\thetab \right)$.  
As already mentioned, famous GSM distributions include the multivariate Student distribution  (when $f_\thetab$ is the pdf of a chi-squared variable), the Pearson type VII distribution (when $f_\thetab$ is a gamma distribution), the generalized Gaussian (when a power of $W$ follows the gamma distribution) 
and the Laplace distribution (when $W^{-1}$ follows an exponential distribution).
It is straightforward to see that GSM are elliptical distributions. However, not all elliptical distributions can be reduced to scale mixtures. For the previous reason, characterization and a way to represent elliptical distributions as GSM are very valuable. 
In \citet{Gomez2006}, conditions are given under which elliptical distributions are GSMs. The issue of finding the corresponding mixing distribution is also addressed. An illustration of these results for generalized Gaussian distributions is given  by \cite{Gomez2008}.

\subsection{Mixtures of High Dimensional GSM}

We consider  finite mixtures of $K$ GSM distributions, with a specific parameterization  to handle potentially high-dimensional observations.
\begin{definition}[Finite mixture]
\label{def:mixture_model}
 $\Yv$ follows a  finite mixture model if its pdf writes as
\begin{equation}
    p(\yv) = \sum \limits_{k=1}^K \pi_k f_k(\yv),
\end{equation}
with $\pi_k \in [0,1]$  the mixing weights summing to one and $f_k$ the pdf of the $k^{\text{th}}$ component.
\end{definition}
A  mixture model where each  $f_k$ is a GSM~\eqref{eq:GSM},
${\cal E}(\mub_k, \Sigmab_k, \thetab_k)$, is referred to as a mixture of ED (MED).  The number of parameters in MED grows quadratically with the dimension $M$ due to the scale matrices $\Sigmab_k$. For large $M$, this can be problematic for estimation. In the Gaussian case,  to reduce the number of parameters, \cite{bouveyron2007} proposed a family of parsimonious Gaussian mixtures, using the eigenvalue decomposition of the covariance matrices. We extend this idea  to  MED by reparameterizing  $\Sigmab_k$ as follows,
\begin{equation}
    \label{eq:repScale}
    \Sigmab_k = \Db_k \Ab_k \Db_k^T,
\end{equation}
where $\Db_k$ is an $M \times M$ orthogonal matrix, which contains the eigenvectors of $\Sigmab_k$ and $\Ab_k$ is an $M \times M$ diagonal matrix, which contains the associated eigenvalues in decreasing order. The key idea of \cite{bouveyron2007} is to consider that each cluster lies in a low-dimensional subspace of dimension $d_k < M$, which can be expressed by assuming that 
\begin{equation}
    \label{eq:decompA}
    \Ab_k = diag (a_{k1}, \ldots, a_{kd_k}, b_k, \ldots, b_k),
\end{equation}
where $a_{k1}, \ldots, a_{kd_k}$ are the $d_k$ largest eigenvalues of $\Sigmab_k$ and $b_k$ is a small negligible value. The  $d_k$ eigenvectors associated to the first $d_k$ eigenvalues $\{a_{k1}, \ldots, a_{kd_k}\}$ define 
a cluster-specific subspace $\mathbb{E}_k$, which 
captures the main cluster shape. The orthogonal subspace is denoted by  
$\mathbb{E}_{k}^{\bot}$.
Let $\widetilde{\Db}_k$ consists of the $d_k$ first columns of $\Db_{k}$ supplemented by $(M-d_k)$ zero columns and $\overline{\Db}_k = (\Db_{k} - \widetilde{\Db}_k)$. It follows that $P_k(\yv) = \widetilde{\Db}_k \widetilde{\Db}_k^T (\yv - \mub_k) + \mub_k$ and $P_{k}^{\bot}(\yv) = \overline{\Db}_k \overline{\Db}_k^T (\yv - \mub_k) + \mub_k$ are the projections of $\yv$ on $\mathbb{E}_k$ and $\mathbb{E}_{k}^{\bot}$ respectively. 
This parameterization allows us to handle high-dimensions in an efficient way. For instance, the quadratic form of the Mahalanobis distance, appearing in the generator $g$ in (\ref{eq:ed}), writes
\begin{eqnarray}
 (\yv - \mub_k) \Db_k \Ab_k^{-1}  \Db_k^T (\yv - \mub_k) 
&=&  (\yv-\mub_k)^T\widetilde{\Db}_k\Ab_k^{-1}\widetilde{\Db}_k^T(\yv-\mub_k) \nonumber \\ & & + (\yv-\mub_k)^T\overline{\Db}_k\Ab_k^{-1}\overline{\Db}_k^T(\yv-\mub_k)   \nonumber \\
& = & \| \mub_k - P_k(\yv)\| _{\widetilde{\Sigmab}_k^{-1}}^{2} +\frac{1}{b_k} \| \yv - P_k(\yv)\|^{2},
\label{eqn:proj}
\end{eqnarray}
where $\|.\|^{2}_{\widetilde{\Sigmab}_{k}^{-1}}$ is the norm defined by $\|\yv\|^{2}_{\widetilde{\Sigmab}^{-1}_{k}} = \yv^{T}\widetilde{\Sigmab}^{-1}_k\yv$ with $\widetilde{\Sigmab}_k^{-1} = \widetilde{\Db}_k\Ab_k^{-1}\widetilde{\Db}_k^T$. Equation~\eqref{eqn:proj} uses the definitions of $P_k$ and $P_k^{\bot}$
and  $\|\mub_k-P_k^{\bot}(\yv)\|^2 = \|\yv-P_k(\yv)\|^2$. The gain comes from the fact that \eqref{eqn:proj} does not depend on $P_k^{\bot}$ and thus does not require the computation of the $(M-d_k)$ latest columns of  $\Db_k$,  the eigenvectors associated to the smallest eigenvalues.   Similarly, determinants can be efficiently computed as  $\text{log}(|\Sigmab_k|) = (\sum_{m=1}^{d_k} \text{log}(a_{km}))+(M-d)\text{log}(b_k)$. 
This parameterization that now only depends on  $\widetilde{\Db}_k$ and not on the complete  $\Db_k$, is indicated by denoting the corresponding ED as 
$\mathcal{HE}_{Md_k} \left(\mub_k, \widetilde{\Db}^*_k, \ab_k, b_k, \thetab_k \right)$, where $\widetilde{\Db}^*_k$ is  matrix $\widetilde{\Db}_k$ with the last zero $M-d_k$ columns omitted.
We then refer to a MED coupled with this  parameterization  as a high-dimensional MED (HD-MED).

\begin{definition}[HD-MED]
A random vector $\Yv \in \Rset^M$ follows a HD-MED distribution  if for all $k\in[1:K]$, the  pdf of the $k^{\text{th}}$ component $f_k$ is an ED $\mathcal{HE}_{\!Md_k} \!\!\left(\mub_k, \widetilde{\Db}^*_k, \ab_k, b_k, \thetab_k \right)$ with reparameterization given by \eqref{eq:repScale} and \eqref{eq:decompA}. 
    We denote 
    
     \noindent   $\Yv \sim \mathcal{MHE}_{M\db} \left((\pi_k,\mub_k, \widetilde{\Db}^*_k, \ab_k, b_k,\thetab_k)_{k=1}^{K}\right),$
    with $\db = \left(d_1, \dots, d_k\right)$, $\ab_k = \left(a_{k1}, \dots, a_{kd_k}\right)$. 
\end{definition}

\section{Dimension reduction with HD-MED}
\label{Sec:Red}

\subsection{Latent variable dimension reduction}

Standard PCA is defined without referring to a probabilistic model. 
Given a set of observations in $\Rset^M$, their $M \times M$ empirical covariance matrix is decomposed into eigenvalues and eigenvectors.
For any  $\yv \in \Rset^M$, a lower dimensional ($d<< M$) representation can then be obtained by considering  its projection to a lower dimensional subspace
$\widehat{\yv} = \hat{\Sigmab}_d^T\yv$ where $\widehat{\Sigmab}_d$ is the matrix containing the $d$ first eigenvectors of the empirical covariance matrix. Its reconstruction in $\Rset^M$, optimal in the sense of the squared reconstruction error, can be obtained with $\widetilde{\yv} = \widehat{\Sigmab}_d\widehat{\yv}$. 
Alternatively, if $\yv$ is assumed to be a realization of a random vector $\Yv \sim  \mathcal{HE}_{Md} \left(\mub, \widetilde{\Db}^*, \ab, b, \thetab\right)$, a low dimensional representation of $\yv$ can be justified using the following latent variable model representation of $\Yv$.
The following proposition generalizes the developments of \cite{archambeau2006,Guo2023} for the Student distribution.

\begin{proposition}[HD-ED latent variable model]
\label{prop:gen_model}
Let $d \leq M-1$, $\Yv \in \Rset^M$,  $\Xv \in \Rset^d$, $\Eb \in \Rset^M$, $W \in \Rset^+$ be random variables,  $\Vb \in \Rset^{M \times d}$ a matrix of linearly independent columns, $\mub \in \Rset^M$ a vector and $f_\thetab$ the pdf of a positive univariate random variable defined by some parameter $\thetab$. 
Assume that
$\Yv = \Vb \Xv + \mub + \Eb$, 
$(\Xv | W=w) \sim {\cal N}(\mathbf{0}_d, w^{-1} \Ib_d)$,
$(\Eb | W=w) \sim {\cal N}(\mathbf{0}_M, b \; w^{-1}\Ib_M)$ and
$W \sim f_\thetab$, 
then, $\Yv \sim \mathcal{HE}_{Md} \left(\mub, \widetilde{\Db}^*, \ab, b, \thetab \right),$ 
with $\Db \Ab \Db^T= b \Ib_M + \Vb \Vb^T$ the eigenvalue decomposition of $b \Ib_M + \Vb \Vb^T$, $\Ab = \operatorname{diag}(a_1, \dots, a_{d}, b, \dots, b)$ the ordered eigenvalues   and $\widetilde{\Db}^*$ the matrix containing the first $d$ eigenvectors as columns.

Additionally, denoting by $\Ub = b \Ib_d + \Vb^T \Vb $, we have, 
\begin{align}
(\Xv | \Yv=\yv, W=w) &\sim {\cal N}( \Ub^{-1}\Vb^T(\yv - \mub), w^{-1} b \Ub^{-1}) \label{eq:condX} \; .
\end{align}
It follows that $\mathbb{E}[\Yv | \Xv=\xv ] =  \Vb\xv +\mub$ and 
$\mathbb{E}[\Xv | \Yv=\yv ] =  \Ub^{-1}\Vb^T(\yv - \mub)\; .$
\end{proposition}

\begin{proof}
    It comes from the first three assumptions that
    $$(\Yv | W=w) \sim {\cal N}(\mub, w^{-1} (\Vb\Vb^T +  b\Ib_M)) $$
    and from  definition \ref{def:GSM} that
    $\Yv  \sim {\cal E}_M(\mub, \Vb\Vb^T +  b\Ib_M, f_\thetab)\; . $
    Since $\Vb\Vb^T$ is of rank $d$, $\Vb\Vb^T +  b\Ib_M$ admits an eigenvalue decomposition $\Db \Ab \Db^T$ as stated.
    Distribution (\ref{eq:condX}) follows from standard Gaussian vectors properties. 
    Using the tower property,
    $\mathbb{E}[\Yv | \Xv=\xv ] = \mathbb{E}[\mathbb{E}[\Yv | \Xv=\xv, W ] ]=  \Vb\xv + \mub$ and 
    $\mathbb{E}[\Xv | \Yv=\yv ] = \mathbb{E}[\mathbb{E}[\Xv | \Yv=\yv, W ] ]=  \Ub^{-1}\Vb^T(\yv - \mub).$
\end{proof}

The previous proposition states that $\yv \in \Rset^M$ from an HD-ED, can also be seen as  originating from a model with a lower-dimensional latent variable $\Xv \in \Rset^d$.  Hence, a natural alternative to the standard PCA projection, is the conditional mean $\widehat{\yv}= \mathbb{E}[\Xv | \Yv=\yv]$, that is 
\begin{equation}
    \label{eq:post_mean}
    \widehat{\yv} = Q(\yv) =  \Ub^{-1} \Vb^T(\yv - \mub),
\end{equation}
and as a reconstruction or an approximation of the original information, the conditional mean  $\check{\yv}= \mathbb{E}[\Yv | \Xv=\widehat{\yv}]$, that is  $\check{\yv} = \Vb \widehat{\yv} + \mub$. Using the previous formulas, lower dimensional representations can  thus be obtained using $\Ub$ and $\Vb$ but when estimating the parameters of the HD-ED, we get estimates for $(a_1, \dots, a_d)$, $b$ and $\widetilde{\Db}^*$. However, using that $\Vb$ is of rank $d$ and $\Db \Ab \Db^T = b \Ib_M + \Vb \Vb^T$, we can set 
\begin{equation}
   \Vb = \widetilde{\Db}^* \sqrt{\operatorname{diag}(a_1, \dots, a_d) - b\mathbf{I}_d}, 
   \label{def:V}
\end{equation}
and deduce $\Ub$ straightforwardly. 
Interestingly, the projection and reconstruction formulas do not directly depend on $f_\thetab$. The dependence is only implicit via the impact of $f_\thetab$ on the estimation of the other parameters.

\subsection{Cluster globally, reduce locally}
The proposed dimension reduction and reconstruction, using the HD-ED model, generalizes to GSM, PPCA
\citep{tipping1999} and robust PPCA \citep{archambeau2006},  the former using Gaussian and the latter Student distributions. Both have been extended to account for potential heterogeneity in data, considering mixtures, with the MPPCA model \citep{tipping1999mix} and its Student-based robust version \citep{archambeau2008}. 
In these models, all clusters are assumed to live in subspaces of the same dimension $d$. In the Gaussian case, an extension, allowing varying reduced dimensions $d_k$ across clusters, have been proposed by \cite{bouveyron2007} with their high-dimensional Gaussian mixture model (HD-GMM). 
The proposed HD-GMM parameterization allows  to handle high-dimensional data in a computationally efficient way. However, it  does not provide an actual lower dimensional representation of the data. While such a reduced-dimensional representation  may often not be  needed, it may be crucial to deal with hardware or software limitations. Originally, HD-GMM have not been designed for dimension reduction or compression but rather for clustering and density estimation in high-dimensional heterogeneous settings.
Our HD-MED model generalizes the covariance matrix decomposition to the scale matrix of a GSM distribution. In addition, we 
describe how it can be further exploited  as a dimension reduction technique. 
As  mixture models, HD-MED can be used for clustering data into $K$ clusters. For any possible observation $\yv$, a  HD-MED model  provides a probability $r_k(\yv)$ that $\yv$ is assigned to cluster $k$ for each $k=1\!:\!K$. Using (\ref{eq:post_mean}), a reduced-dimension representation $\widehat{\yv}_k$ of $\yv$, for each of the $K$ different subspaces, is denoted by $\widehat{\yv}_k = Q_k(\yv)$ and given by,
\begin{equation}
    \label{eq:post_meank}
    \widehat{\yv}_k = Q_k(\yv) = \Ub_k^{-1} \Vb_k^T(\yv - \mub_k),
\end{equation}
while its  reconstruction $\check{\yv}_k$ in the original space is given by $\check{\yv}_k = \Vb_k \widehat{\yv}_k + \mub_k\;. $
In practice, it is reasonable to use as a reduced-dimension representation of  $\yv$ only the one corresponding to the most probable cluster $k$, {\it i.e.} with the highest $r_k(\yv)$. 
In this setting, HD-MED acts as a divide-and-conquer paradigm by initially clustering the data into $K$ clusters and then performing cluster-specific data reduction.  
The divide step allows a much more effective reduction than if a single subspace was considered, while in the conquer step, 
little information is lost, as for any new observation $\yv$, cluster assignment probabilities $r_k(\yv)$ can be straightforwardly computed to decide on the best reduced representation to be used.
However, it is important to keep track of clustering information for each observation. The reduced representations cannot be pooled back altogether, as they are likely to become impossible to distinguish across clusters. 
Also, as summarized in Algorithm~\ref{alg:DC} and illustrated in Section \ref{sec:MRF}, each reduced cluster may have to be  processed  separately but this additional cost is negligible compared to the hardware and software gain of a more efficient representation. 

The Expectation-Maximization (EM) algorithm (\cite{dempster1977}) that iteratively  maximizes the conditional expectation of the complete-data log-likelihood,  is commonly used to estimate finite mixture models. 
Gradient algorithms can also be used but have not been yet very popular, see a more detailed discussion in Appendix Section 2.
The number of clusters $K$, and their respective inner dimension $d_k$ are hyperparameters that need to be  tuned prior to the EM steps. As a simple solution, we use the Bayesian Information Criterion (BIC) to tune  $K$ and the $d_k$'s at the same time. 
In contrast to MPPCA solutions which assume the same subspace dimension $d$ for all clusters, the possibility to handle different $d_k$'s and to allow non Gaussian cluster shapes is important for the target applications involving datasets that are very large.  Using different dimensions across clusters is likely to yield a more efficient reduced representation of the data as illustrated in Section \ref{sec:MRF}. 

The standard, or batch EM algorithm needs all the dataset to be loaded in a fast access memory (\textit{e.g} the RAM) which is often limited as this kind of memory is expensive. In the case of large-scale dataset, the RAM is often overloaded and supported by a slow-memory which makes EM iterations  very slow. 
Batch sizes are then limited by resource constraints, so that very large data sets need either to be downsampled or to be handled in an incremental manner. Incremental versions of EM 
exist and can be adapted to our setting. 
In section~\ref{sec:oem}, we  provide a way to deal with this large-scale case by using an online version of  EM.

\section{Online Learning of High Dimensional Mixtures of Elliptical Distributions}
\label{sec:oem}
\subsection{Online EM, main assumptions}
When the data volume is too large, EM becomes slow because of multiple data transfer between the RAM and the  computer storage. A way to handle large volumes is to  use  online learning. Online learning refers to procedures dealing with data acquired sequentially. Online  variants of EM, among others, are described in \cite{Cappe2009,Maire:2017aa,Karimi:2019ab,Karimi:2019aa,Fort:2020aa,Kuhn:2020aa,Nguyen:2020ac,Dieuleveut2023}. As an archetype of such algorithms, we  consider the online EM of \cite{Cappe2009} which belongs to the family of stochastic approximation  algorithms (\cite{borkar2009,Dieuleveut2023}). This algorithm has been theoretically studied and extended. However, it is designed only for distributions that admit a data augmentation scheme yielding a  complete likelihood of the exponential family form, see \eqref{eq: A1} below.  This case is already very broad, including  {\it e.g.} Gaussian, gamma, Student distributions and mixtures of those.  We recall  the main assumptions required and the online EM iteration, based on a latent variable formulation. 

Assume $\left(\Yv_{i}\right)_{i=1}^{N}$ is a sequence of $N$ {\it i.i.d.} replicates of a random variable $\Yv \in  {\cal Y} \subset \Rset^{M}$, observed  one at a time. Extension to successive mini-batches  is  straightforward (\cite{Nguyen:2020ac}). In addition, $\Yv$ is assumed to be the visible part of $\left(\Yv,\Zv\right)$, where $\Zv\in\Rset^{l}$ is a latent variable, {\it e.g.} the unknown component label in a mixture or a mixing weight in a GSM formulation, and $l\in\Nset$. For $i \in [1:N]$ then, each $\Yv_{i}$ is the visible part of $\left(\Yv_{i}, \Zv_{i}\right)$. Suppose $\Yv$ arises from some data generating process (DGP) characterised by a probability density function  $f\left(\yv;\sbmm{\Theta}_0\right)$, with unknown parameters $\boldsymbol{\Theta}_0\in\mathbb{T} \subseteq \Rset^{p}$, for $p\in\Nset$. 

Using the sequence $\left(\Yv_{i}\right)_{i=1}^{N}$, the  method of \cite{Cappe2009} allows to sequentially estimate $\boldsymbol{\Theta}_{0}$ provided the  following assumptions are met:

\begin{itemize}
\item [{(A1)}] The complete-data likelihood for $(\Yv, \Zv)$
is of the exponential family form: 
\begin{equation}
f_{c}\left(\yv, \zv;\boldsymbol{\Theta}\right) = h\left(\yv, \zv\right)\exp\left\{ \left[\mathbf{s}\left(\yv, \zv\right)\right]^{\top}\boldsymbol{\phi}\left(\boldsymbol{\Theta}\right)-\psi\left(\boldsymbol{\Theta}\right)\right\} \text{,}\label{eq: A1}
\end{equation}
with $h : \Rset^{M+l} \rightarrow \left[0,\infty\right)$, $\psi : \Rset^{p} \rightarrow \Rset$, $\mathbf{s} : \Rset^{M+l} \rightarrow \Rset^{q}$, $\boldsymbol{\phi} : \Rset^{p}\rightarrow\Rset^{q}$, for $q \in \Nset$.
\item [{(A2)}] The function
\begin{align}
\bar{\mathbf{s}}\left(\yv;\boldsymbol{\Theta}\right) & =\mathbb{E}\left[\mathbf{s}\left(\Yv,\Zv\right)|\Yv = \yv ; \boldsymbol{\Theta}\right]\label{eq: A2}
\end{align}
is well-defined for all $\yv$ and $\boldsymbol{\Theta}\in\mathbb{T}$,
where $\mathbb{E}\left[\cdot|\Yv=\yv;\boldsymbol{\Theta}\right]$
is the conditional expectation when $\Xv$ arises from the DGP characterised by $\boldsymbol{\Theta}$.
\item [{(A3)}] There is a convex $\mathbb{S}\subseteq\Rset^{q}$,
 satisfying:
(i) for all $\gamma \in \left(0,1\right)$, $\mathbf{s} \in \mathbb{S}$, $\yv\in {\cal Y}$, and $\boldsymbol{\Theta}\in\mathbb{T}$,
$\left(1-\gamma\right)\mathbf{s}+\gamma\bar{\mathbf{s}}\left(\yv;\boldsymbol{\Theta}\right)\in\mathbb{S}\text{;}$ and 
(ii) for any $\mathbf{s}\in\mathbb{S}$, the function 
$Q\left(\mathbf{s};\boldsymbol{\Theta}\right)=\mathbf{s}^{\top}\boldsymbol{\phi}\left(\boldsymbol{\Theta}\right)-\psi\left(\boldsymbol{\Theta}\right)$
has a unique global maximizer on $\mathbb{T}$ denoted by
\begin{equation}
\bar{\boldsymbol{\Theta}}\left(\mathbf{s}\right)=\underset{\boldsymbol{\Theta}\in\mathbb{T}}{\arg\max}\;Q\left(\mathbf{s};\boldsymbol{\Theta}\right)\text{.}\label{eq: def theta}
\end{equation}
\end{itemize}

Let $\left(\gamma_{i}\right)_{i=1}^{N}$ be a sequence of learning
rates in $\left(0,1\right)$ and $\boldsymbol{\Theta}^{\left(0\right)}\in\mathbb{T}$
 an initial estimate of $\boldsymbol{\Theta}_{0}$. For each $i\in [1:N]$,
the online EM of \cite{Cappe2009} proceeds by computing 
\begin{align}
\mathbf{s}^{\left(i\right)}&=\gamma_{i}\bar{\mathbf{s}}(\yv_{i};\boldsymbol{\Theta}^{(i-1)})+\left(1-\gamma_{i}\right)\mathbf{s}^{\left(i-1\right)}\text{,}\label{eq: S up}\\
 \mbox{and} \qquad
\boldsymbol{\Theta}^{\left(i\right)}&=\bar{\boldsymbol{\Theta}}(\mathbf{s}^{\left(i\right)})\text{,}\label{eq: Theta up}
\end{align}
where $\mathbf{s}^{\left(0\right)}=\bar{\mathbf{s}}(\yv_{1};\boldsymbol{\Theta}^{\left(0\right)})$.
Theorem 1 of \cite{Cappe2009} shows that  the sequence  $(\boldsymbol{\Theta}^{\left(i\right)})_{i=1:N}$ of estimators of $\boldsymbol{\Theta}_{0}$  satisfies a convergence result to stationary points of the likelihood (\textit{cf.} \cite{Cappe2009} for a more precise statement).

\subsection{Online EM for HD-ED}
We now derive the online EM for HD-ED. The extension to HD-MED mixtures is straightforward and is detailed in \cite{NguyenForbes2021} or in the supplementary Section 3.
The weight distribution $f$ in the GSM formulation \eqref{eq:GSM} is assumed to belong to the exponential family. This case may seem restrictive but it  encompasses a number of ED as the Gaussian, Student, Normal Inverse Gamma with no skewness \textit{etc.} distributions.

\begin{proposition}[HD-ED exponential form]

Let $\Yv$ be a HD-ED distributed variable,  $\Yv \sim \mathcal{HE}_{Md} (\mub, \widetilde{\Db}^*, \ab, b, \thetab)$, and  $W$ a weight variable with pdf $f_\thetab$. The set of parameters is denoted by $\Thetab = (\mub, \widetilde{\Db}^*, \ab, b, \thetab)$, with $\widetilde{\Db}^*$ defined by its column vectors $\widetilde{\Db}^* = \left[\db_1, \dots, \db_d\right]$.
If $W$ belongs to the exponential family, {\it i.e.}
$f_{\thetab}(w) = h_w(w) \exp \left[\sb_w(w)^T  \phib_w(\thetab) - \psi_w(\thetab) \right],$
the complete data likelihood
\begin{equation}
    \label{eq:comp_like}
    f_c(\yv, w; \Thetab) = f_{\thetab}(w)\;  \mathcal{N}_M(\yv; \mub, w^{-1}\Db \Ab \Db^T),
\end{equation}
can be expressed in an exponential family form  (\ref{eq: A1}) with
\begin{equation}
    \label{eq:exp_form}
    \begin{aligned}
    \sb(\yv, w) &= \begin{bmatrix}
           w \yv \\
           w \operatorname{vec}(\yv\yv^T) \\
           w \yv^T\yv \\
           w  \\
           \sb_w(w) \\
         \end{bmatrix}, \hspace{.5cm}
    \phib(\Thetab) = \begin{bmatrix}
          \sum \limits_{m=1}^d \left( \frac{1}{a_m} - \frac{1}{b}\right) \db_m \db_m^T \mub + \frac{1}{b}\mub \\
          \frac{1}{2} \sum \limits_{m=1}^{d} \left( \frac{1}{b} - \frac{1}{a_m}\right) \operatorname{vec}(\db_m \db_m^T) \\
          -\frac{1}{2b} \\
          \frac{1}{2} \sum \limits_{m=1}^d \left( \frac{1}{b} - \frac{1}{a_m}\right) \mub^T\db_{m}\db_{m}^T\mub - \frac{1}{2b}\mub^T \mub \\
          \phi_w(\thetab) \\
         \end{bmatrix}, \\
         \psi(\Thetab) &= \frac{1}{2} \sum \limits_{m=1}^{d} \log a_m + \frac{M-d}{2} \log b + \psi_w(\thetab),
    \end{aligned}
\end{equation}
where $\operatorname{vec}$ denotes the vectorization operator (\cite{Schott2016}).
\end{proposition}

\begin{proof}
    The proof is detailed in supplementary Section 4.
\end{proof}

The online EM algorithm (OEM) consists, as the batch EM, of two steps, the computation of the sufficient statistics~\eqref{eq: A2} and the maximization of the likelihood~\eqref{eq: def theta}. For the first step, we need to compute $\overline{\sb} \left(\yv; \Thetab\right) = \mathbb{E}[\sb(\Yv, \Zv) | \Yv = \yv; \Thetab]$. This quantity requires to compute the following expectations $\mathbb{E}[W | \Yv = \yv; \Thetab]$, and $\mathbb{E}[s_w(W)| \Yv = \yv; \Thetab]$.

\begin{proposition}[Expectations of sufficient statistics]
Let  $\Yv \sim \mathcal{HE}_{Md}(\mub, \widetilde{\Db}^*, \ab, b, \thetab)$, $\Thetab = \left(\mub, \widetilde{\Db}^*, \ab, b, \thetab\right)$, and  $W$ the mixing variable $W \sim f_{\thetab}$. Then $\Yv$ has density~\eqref{eq:ed} with generator $g$ and Mahalanobis distance $u = (\yv-\mub)^T \Db\Ab^{-1}\Db^T(\yv-\mub)$ defined as in \eqref{eqn:proj}. With  $g'$ denoting the derivative of $g$, it follows that
\begin{align}
    \mathbb{E}\left[W | \Yv=\yv; \Thetab\right] &= -\frac{2}{(2 \pi)^{M/2}} \frac{g'(u)}{g(u)}. \label{eq:expectation_T} 
\end{align} 
\end{proposition}

\begin{proof}
Equation~\eqref{eq:expectation_T} results from
 equations \eqref{eq:ed} and \eqref{eq:GSM}, which lead to
 
  $  \mathbb{E}\left[W | \Yv=\yv; \Thetab\right] = \frac{1}{(2\pi)^{M/2} g(u)}\int \limits_{\Rset^+} w \; f_{\thetab}(w) \exp (- \frac{w}{2}u) \diff w = -\frac{2}{(2 \pi)^{M/2}} \frac{g'(u)}{g(u)}.
  $
\end{proof}

 In contrast, there is no general formula for the expectation of $s_w(W)$, which depends on the mixing distribution $f_{\thetab}$. Once the expectation of the sufficient statistics in~\eqref{eq: A2}  is computed, we can update it following~\eqref{eq: S up}.

The next OEM step is the maximization step described in~\eqref{eq: def theta}, which gives an estimation of the parameters at each iteration. The solution for $\thetab$ varies with $f_\theta$, but solutions for $\mub$, $\Ab$, and $\widetilde{\Db}^*$ can be derived as follows. Let  $\bar{\Thetab}(\sb)$ be defined as the unique maximizer of function $Q(\sb, \Thetab) = \sb^T \phib(\Thetab) - \psi(\Thetab)$ with $\sb$ a vector that matches the definition and dimension of $\phib(\Thetab)$ in~\eqref{eq:exp_form}, and can be conveniently written as
\begin{equation}
    \sb = \begin{bmatrix}
          \sb_{1} \\
          \operatorname{vec}(\Sb_{2}) \\
           s_{3} \\
           s_{4}  \\
           s_{5} \\
         \end{bmatrix},
\end{equation}
with $\sb_{1}$ an $M$-dimensional vector, $\Sb_{2}$ an $M \times M$ matrix, and  $s_{3}$, $s_{4}$, $s_{5}$ three scalar values. Parameters are  updated by maximizing $Q$ with respect to $\Thetab$. $\bar{\Thetab}(\sb)$ is  defined as the root of the first-order condition involving $\Jb_{\phib}(\Thetab) = \frac{\partial \phib}{\partial \Thetab}$  the Jacobian of $\phi$,
\begin{equation}
    \label{eq:root}
    \Jb_{\phib}(\Thetab)\sb - \frac{\partial \psi}{\partial \Thetab}(\Thetab) = \mathbf{0}.
\end{equation}
 Computing gradients leads to $\bar{\Thetab}(\sb) = (\overline{\mub}(\sb), \overline{\widetilde{\Db}^*}(\sb), \overline{\Ab}(\sb), \overline{\thetab}(\sb)$, where  $\overline{\mub}(\sb)$, and $\overline{\Ab}(\sb)$ are closed form, and $\overline{\widetilde{\Db}^*}(\sb)$ can be found using Riemannian optimization.
\begin{proposition}[Maxima]
We use an ECM-like procedure (\cite{meng1993}) by optimizing  parameters separately and incorporating them during the optimization of each other parameter. $(\mub, \Ab)$ are easily optimized using~\eqref{eq:root} and computing the gradients gives
\begin{align}
    \overline{\mub} &= \frac{\sb_{1}}{s_{4}}, \\
    \overline{a}_m &= \db_m^T \left(\Sb_{2} + s_{4}  \overline{\mub}  \overline{\mub}^T - 2  \overline{\mub} \sb_{1}^T \right) \db_m \quad \text{for } m \in [1:d], \\
    \overline{b} &= \frac{1}{M - d} \left(s_{4}  \overline{\mub}^T  \overline{\mub} + s_{3} - 2  \overline{\mub}^T \sb_{1} - \sum \limits_{m=1}^d \overline{a}_m \right).
\end{align}
Maximizing over $\widetilde{\Db}^*$ has to take into account that $\widetilde{\Db}^* \in St(M, d)$, the Stiefel manifold of the $M \times d$ matrices. Using the expressions of $\mub$ and $\Ab$, omitting parts  on $\thetab$ we have, 
\begin{equation}
    \label{eq:max_d}
    \overline{\widetilde{\Db}^*} = \argmax \limits_{\widetilde{\Db}^*\in St(M, d)} \sum \limits_{m=1}^d \left( \frac{1}{\overline{a_{m}}} -\frac{1}{\overline{b}}\right)\db_{m}^T \left( 2 \overline{\mub} \sb_{1}^T - \Sb_{2} - s_{4} \overline{\mub}  \overline{\mub}^T \right)\db_{m}.
\end{equation}
For $\bar{\thetab}$, a general closed-form expression is not available, but if there are no particular constraints it results from solving the following equation
\begin{equation}
    \label{eq:root_eq}
    s_{5}\frac{\partial \phib_w}{\partial \thetab}(\bar{\thetab}) - \frac{\partial \psi_w}{\partial \thetab}(\bar{\thetab}) = \mathbf{0}.
\end{equation}
\end{proposition}

\begin{proof}
    We compute the gradients and use \eqref{eq:root} for $\mub$ and $\Ab$, however for the   $\widetilde{\Db}^*$ we only plug-in the optimized values $\overline{\mub}$, and $\overline{\Ab}$  in \eqref{eq:comp_like} and solve
    \begin{equation}
        \argmax \limits_{\widetilde{\Db}^*\in St(M, d)} \sb^T \phib \left( \overline{\mub}, \widetilde{\Db}^*,\overline{\Ab}, \bar{\thetab} \right) - \psi \left( \overline{\mub}, \widetilde{\Db}^*, \overline{\Ab}, \bar{\thetab} \right),
    \end{equation}
    which does not depend on the value of $\bar{\thetab}$.
    \end{proof}
    In practice \eqref{eq:max_d} is solved using a Riemannian optimization framework in the setting where $M >> d$ \citep{wen2013}.

\begin{algorithm}[h]\scriptsize
    \caption{\small Divide \& Conquer high-dimensional matching for MRF reconstruction}
    \label{alg:DC}
    \begin{algorithmic}[1]
     \item[ ] {\bf Input} Dictionary of (signal, parameters) pairs ${\cal D}_f = \{\yv_i, \tv_i\}_{i=1:N}$, $N>>1$, $\tv_i \in \Rset^L$, $\yv_i \in \Rset^M$, $M >> 1$.\\
    $\qquad$ In vivo acquired signals $\{\widetilde{\yv}_j\}_{j=1:\tilde{N}}$, $\widetilde{\yv}_j \in \Rset^M$. 
     \begin{tcolorbox}[colback=orange!10,colframe=orange!50!black, boxrule=0.01pt,boxsep=2pt, top=1pt, bottom=1pt, title style={halign=flush right,before upper=\strut}]
        \vspace{1pt}
        \STATE {\bf Reduced dimension representation of the dictionary:} $\{\hat{\yv}_i, \tv_i, r.(\yv_i)\}_{i=1:N} $
        \begin{tcolorbox}[colback=cyan!10,colframe=cyan!50!black, boxrule=0.01pt,boxsep=2pt, top=1pt, bottom=1pt, title style={halign=flush right,before upper=\strut}]
            \item[1.1] {\bf Online HD-MED inference from $\{\yv_i\}_{i=1:N}$:} $K$ clusters, $d_k<M$ for $k=1:K$  $\quad \Longrightarrow$  cluster assignment probabilities and cluster-wise projections $(\rb,\Qb) =\{r_k(\cdot), Q_k(\cdot)\}_{k=1:K}$
        \end{tcolorbox}

        \begin{tcolorbox}[colback=red!10,colframe=red!50!black, boxrule=0.01pt,boxsep=2pt, top=1pt, bottom=1pt, title style={halign=flush right,before upper=\strut}]
            \item[1.2] {\bf Cluster-wise fingerprint reductions:}  $\{\yv_i\}_{i=1:N}, \rb, \Qb \quad \Longrightarrow \quad 
            \{ \widehat{\yv}_i=Q_k(\yv_i), i \in I_k\}\quad $ with
            $\quad I_k=\{i, \; s.t  \; k=\arg\max\limits_\ell r_\ell(\yv_i)\}$, for $k=1:K$
        \end{tcolorbox}
    \end{tcolorbox}

        \begin{tcolorbox}[colback=green!10,colframe=cyan!50!black, boxrule=0.01pt,boxsep=2pt, top=1pt, bottom=1pt, title style={halign=flush right,before upper=\strut}]
            \STATE {\bf Cluster-wise matching of acquired signals:} 
             \begin{tcolorbox}[colback=cyan!10,colframe=cyan!50!black, boxrule=0.01pt,boxsep=2pt, top=1pt, bottom=1pt, title style={halign=flush right,before upper=\strut}]
            \item[2.1] {\bf Cluster-wise invivo signal reductions:} Use learned $(\rb, \Qb)$ from \textit{step 1.1}  to obtain $\{Q_k(\widetilde{\yv}_j), j \in \tilde{I}_k\}\quad $ with
            $\quad \tilde{I}_k=\{j, \; s.t  \; k=\arg\max\limits_\ell r_\ell(\widetilde{\yv}_j)\}$, for $k=1:K$ 
            \end{tcolorbox}
            \begin{tcolorbox}[colback=red!10,colframe=red!50!black, boxrule=0.01pt,boxsep=2pt, top=1pt, bottom=1pt, title style={halign=flush right,before upper=\strut}]
            \item[2.2]{\bf Matching:} For $k=1:K$, for $j \in \tilde{I}_k$, determine $i(\tilde{\yv}_j) =  \arg\min\limits_{i \in I_k}d(Q_k(\widetilde{\yv}_j), \widehat{\yv}_i)$ and set $\widetilde{\tv}_j = \tv_{i(\tilde{\yv}_j)}$
            \end{tcolorbox}
            
        \end{tcolorbox}
        \item[ ] {\bf Return}  Matched tissue properties 
        $ \{\widetilde{\tv}_j\}_{j=1:\tilde{N}}$, $\widetilde{\tv}_j \in \Rset^L$
    \end{algorithmic}
\end{algorithm}

\section{Application to magnetic resonance fingerprinting (MRF) reconstruction}
\label{sec:MRF}
MRF  provides multiple tissue parameter maps from shorter acquisition times, thanks to  the simultaneous application of transient states excitation and highly undersampled $k$-space read-outs. These two aspects have a combined impact on  acquisition times and  image reconstruction accuracy.  More undersampling allows more parameter estimations  but is also  responsible for larger undersampling errors, noise and artifacts, reducing map reconstruction accuracy.

In  some earlier work \citep{Oudoumanessah2025}, the Gaussian version of our procedure, referred to as HD-GMM, was evaluated, for the reconstruction, in an ideal setting, of fully sampled acquisitions, targeting  $L=3$ parameters. In this scenario, it was reported that HD-GMM, coupled with the online EM algorithm, achieved results comparable to full dictionary matching while significantly reducing reconstruction times. However, as we report in this section, HD-GMM performance degrades when dealing with largely undersampled acquisitions and  becomes inadequate for estimating six parameters. Performance can be maintained by considering an elliptical version of our procedure (HD-MED) less sensitive to outliers.

In this section, we thus show how we can accurately reconstruct six parameter maps from \textit{in vivo} undersampled acquisitions (Section \ref{sec:acq}) by leveraging an extensive high-resolution dictionary (Section \ref{sec:dic}) and the HD-MED model. Algorithm~\ref{alg:DC} provides a schematic summary of our procedure. Figure~\ref{fig:pipeline} provides an illustration of our matching strategy or step 2 in Algorithm \ref{alg:DC}, once the HD-MED model has been estimated. Each acquired signal is first assigned to one of the learned clusters, and reduced accordingly to be then matched to the best dictionary signal in the corresponding cluster. It leads to improvements in both memory management complexity and reconstruction speed for parameter maps when compared to traditional dictionary matching, see Section \ref{sec:times}. 

We compare two instances of HD-MED, namely HD-GMM and HD-STM, where the mixture components are respectively set to Gaussian and Student distributions. The computations for the corresponding OEM are detailed in Section 5 of the Appendix.
Additional computations and results for mixtures of Laplace distributions (HD-LM) are given in Appendix Section 6.
All experiments are performed with a Python code using the JAX library \citep{jax2018github} with Nvidia H100-80gb GPU for training, and V100-32gb GPU for inference except for the dictionary generation part that is done with a mix of Matlab and Python. 

\begin{figure}
    \centering
    \includegraphics[width=0.7\linewidth]{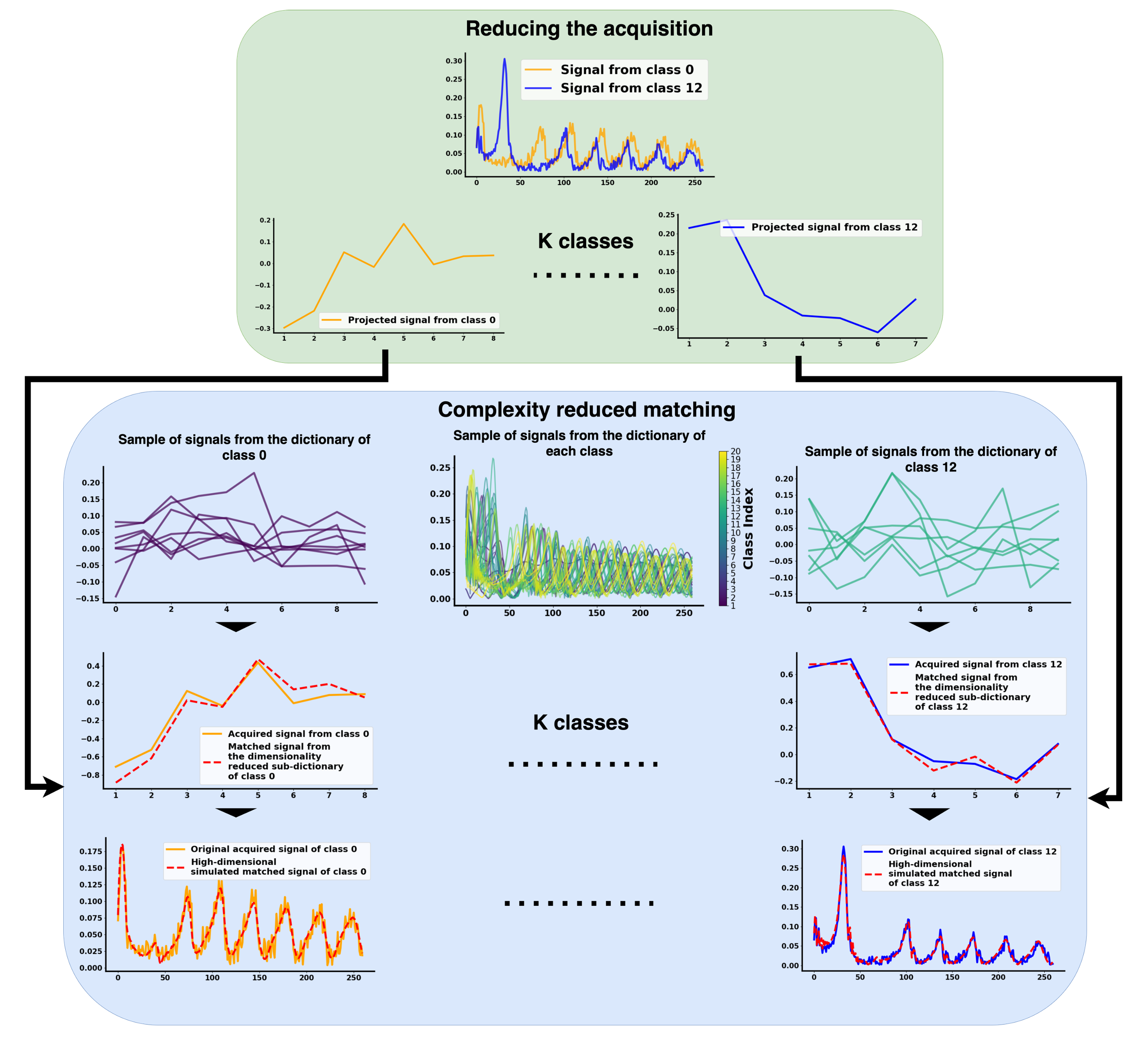}
    \caption{HD-MED-based low-dimensional signal matching step (step 2 of Algorithm \ref{alg:DC}) Green block: each acquired signal is assigned to one of the $K$  HD-MED learned clusters and reduced accordingly using the assigned cluster projection. 
    Blue block: reduced acquired signals are then matched with their closest reduced simulated counterpart in the cluster.}
    \label{fig:pipeline}
\end{figure}

\subsection{Undersampled MRF acquisitions}
\label{sec:acq}
\textit{In vivo} acquisitions were conducted on 6 healthy volunteers (28$\pm$5.5 years old, 3 males and 3 females) using a 32-channel head receiver array on a Philips 3T Achieva dStream MRI at the IRMaGe facility (MAP-IRMaGe protocol,  NCT05036629).  This study was approved by the local medical ethics committee and informed consent was obtained from each volunteer prior to image acquisition. The imaging pulse sequence was based on an IR-bSSFP acquisition. 260 repetitions were acquired following the parameters proposed in \cite{coudert2024}. The acquisitions were performed using quadratic variable density spiral sampling (12 interleaves out of 13), matrix size=192x192x(4-5), voxel size=1.04x1.04x3.00~mm$^3$ for a total scan duration of 2 minutes per slice. While the acquisition time may appear long for an MRF context compared to \cite{gu2018}, where full T1 and T2 maps are generated in 2 minutes, the longer sequences in \cite{coudert2024} account for vascular parameters without requiring contrast agents.

Spiral sampling in MRI involves acquiring data in a spiral trajectory through k-space, covering the center first and gradually moving outward, which allows for faster data acquisition and more efficient use of the scanner's gradients. However, this method leads to undersampling noise, manifesting as artifacts, because it captures less information about the image, reducing the ability to accurately reconstruct fine details \citep{korzdorfer2019}.
It results in a number $\tilde{N}_s\approx 140,000 - 180,000 $ of {\it in vivo} MR signals per subject $s$, each of dimension $M=260$, and a total number $\tilde{N}$ of approximately $\tilde{N}=960,000$ signals to be matched for the reconstruction of $L=6$ parameter maps for 6 subjects.

\subsection{MRF dictionary}
\label{sec:dic}
The MRF sequence used here is a bSSFP-derived sequence. This type of sequence has been preferred because of its sensitivity to local frequency distributions related to microvascular structures in the imaging voxel.
The MRF dictionary is generated using the approach described in \cite{coudert2024}. First, a base dictionary ${\cal D}_{f_0}$ is simulated using Bloch equations for different combinations of ($T_1$, $T_2$, $B_1$, $\delta \! f$) sampled over a 4-dimensional regular grid. Simulations were made at a magnetic field strength of $3.0$T, on a regular parameter grid made of $20$ $T_1$ values (from $200$ to $3500ms$), $20$ $T_2$ values (from $10$ to $600ms$), $10$ $B_1$ values (from $0.7$ to $1.2$) and $100$ frequency offset $\delta \! f$ values (from $-50$ to $49$ Hz with an increment of 1 Hz), keeping only signals for which \tun$>$\tdeux, resulting in an initial $390,000$ entries dictionary. These Bloch simulations are obtained using an in-house mix of Python and Matlab code based on initial implementation by B. Hargreaves \citep{Bloch_Simulator}. The relaxometric parameters vary within the dictionary to allow their estimation. $B_1$ is varied to ensure the realism of the dictionary, given the sequence's sensitivity to this parameter, which could otherwise bias the estimation of the other parameters. Finally, varying $\delta \! f$ values are necessary to compute microvascular contributions, as detailed below.

To build signals that capture additional microvascular network's blood volume (CBV) and mean vessel radius (R) information, we follow the construction of \cite{coudert2024}. Microvascular network segmentations are used to pre-compute 2500 different intra-voxel frequency distributions centred at $\delta \! f$ values. A CBV and R value characterize  each distribution. Using these pre-computed distributions, new signals are obtained by summing signals from the base dictionary weighted according to each frequency distribution. The result is an expanded 6-dimensional dictionary ${\cal D}_{f}$ of almost $N=400,000,000$ signals of dimension $M=260$, encoding for  ($T_1$, $T_2$, $B_1$, $\delta \! f$, $CBV$, $R$). Figure~\ref{fig:convol} illustrates how this operation changes a signal in the base dictionary into 2 new signals depending on $CBV$ and $R$ values.

\begin{figure}
    \centering
    \includegraphics[width=0.5\linewidth]{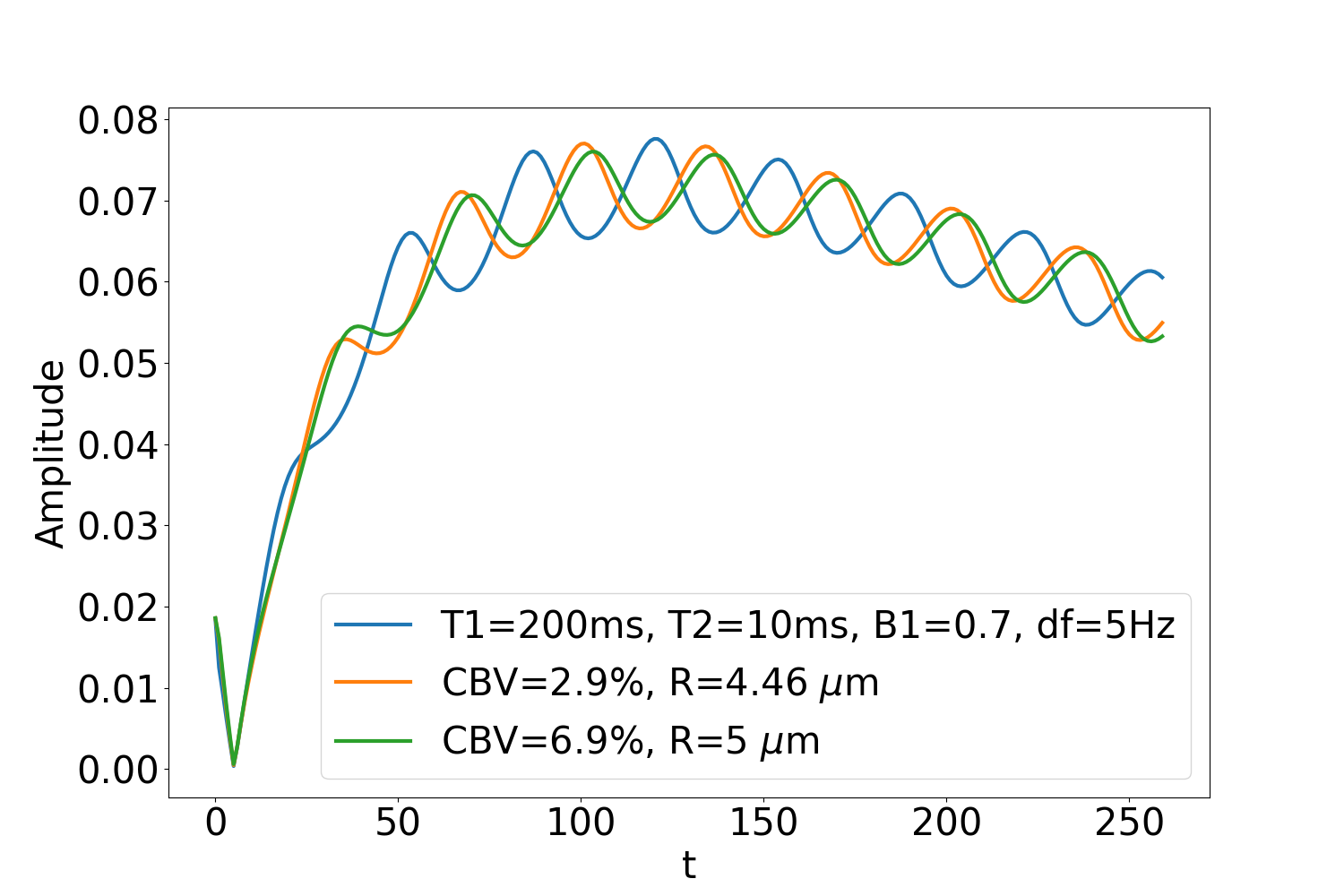}
    \caption{Illustration of the base dictionary expansion. The blue signal corresponding to a non-blood voxel with associated $T_1$, $T_2$, $B_1$ and $\delta \! f$ values. By convoluting this signal along the $\delta \! f$ dimension with two different frequency distributions corresponding to different values of $CBV$ and $R$, the resulting orange and green signals are then obtained for voxels with the same $T_1$, $T_2$, $B_1$ and $\delta \! f$ values as the blue signal but with different $CBV$ and $R$ values.}
    \label{fig:convol}
\end{figure}

\subsection{Model selection and initialization}

In practice, the online EM algorithm is run in batches \citep{Nguyen:2020ac} of size $B=2048$ signals at each iteration $r$ where $r\in [1:R]$ for typically $R = \lceil N/B\rceil$. 
When $N$ is huge, the $N$ signals are read only once in a single epoch, while for more moderate $N$, several epochs, {\it e.g.} 10, can be run with shuffling, making use of the fact that samples are available and can be revisited.
Regarding the learning rates,  $(\gamma_r)_{r=1}^R$, they are set to $\gamma_r = (1-\epsilon) r^{-0.5 -\epsilon'}$ with $\epsilon=10^{-10}, \epsilon'=10^{-1}$ and Polyak-Ruppert averaging, following theoretical recommendations. In particular, for this choice,  the MSE converges at the optimal rate and the covariance in the CLT is minimal, see \cite{Lauand2024} for additional  results and a recent literature survey.

The only two other hyperparameters that need to be set are the number $K$ of components  and the desired reduced dimension vector $\db$. 
Once initialized, these quantities are considered as fixed and not updated during the online procedure.
For mixture models, $K$ can be selected using the Bayesian Information Criterion (BIC) \citep{schwarz1978}, which requires running multiple models with varying $K$ values and finding the elbow or the minimum of the BIC curve. The desired dimension vector $\db$ is automatically set during an initialization step, which makes  $K$ the only parameter that the user needs to set prior to running the OEM algorithm.

The initialization step includes determining  the reduced dimension vector $\db$. We randomly choose a subset of the dictionary that fits into memory and run a batch full covariance EM algorithm to determine a first estimation of the $\Sigmab_k$'s. 
In practice, a subset of 10,000 signals was used.
The covariance estimations are  then decomposed into eigenvalues and eigenvectors to determine $\Db_k$'s and $\Ab_k$'s. For each $k$, $d_k$ is then determined by applying a scree plot to the eigenvalues using the kneedle algorithm \citep{satopaa2011}, $\ab_k$ is then initialized to the first $d_k$ eigenvalues, and the last eigenvalues are averaged to initialize $b_k$. This kind of initialization of $\Db_k$ and $\Ab_k$ is called spectral initialization and is similar to the one proposed by \cite{hong2021}, proving to be relatively stable compared to other types of initializations.

In Figure~\ref{fig:select} (left), we show the different BIC curves for HD-GMM and HD-STM models trained on the dictionary of signals with  $K$ varying  from $5$ to $80$. The elbows of the curves are found at $K=30$ for HD-GMM, and  at $K=25$ for HD-STM.  Figures~\ref{fig:select} (middle and right) show the different vectors $\db$ obtained for varying $K$. The larger the points, the higher the number of components having the reduced dimension indicated on the $y$-axis. Most of these dimensions are between $5$ and $40$, meaning that the original $M=260$ signal dimension can be reduced by a factor of approximately $7$ to $10$ leading to a reduction from $13$ TB to approximately $1.5$ TB of signals stored in memory.  To check, how much is lost in this reduction, we compute the root mean square error (RMSE) between the dictionary $\yv_i$ signals and their reconstructions $\check{\yv}_i$ from their reduced representations (\ref{eq:post_meank}),
$ RMSE = \sqrt{\frac{1}{N}\sum_{i=1}^N (\yv_i -\check{\yv}_i)^2} \;.$

\begin{figure}[h]
        \includegraphics[width=0.32\textwidth]{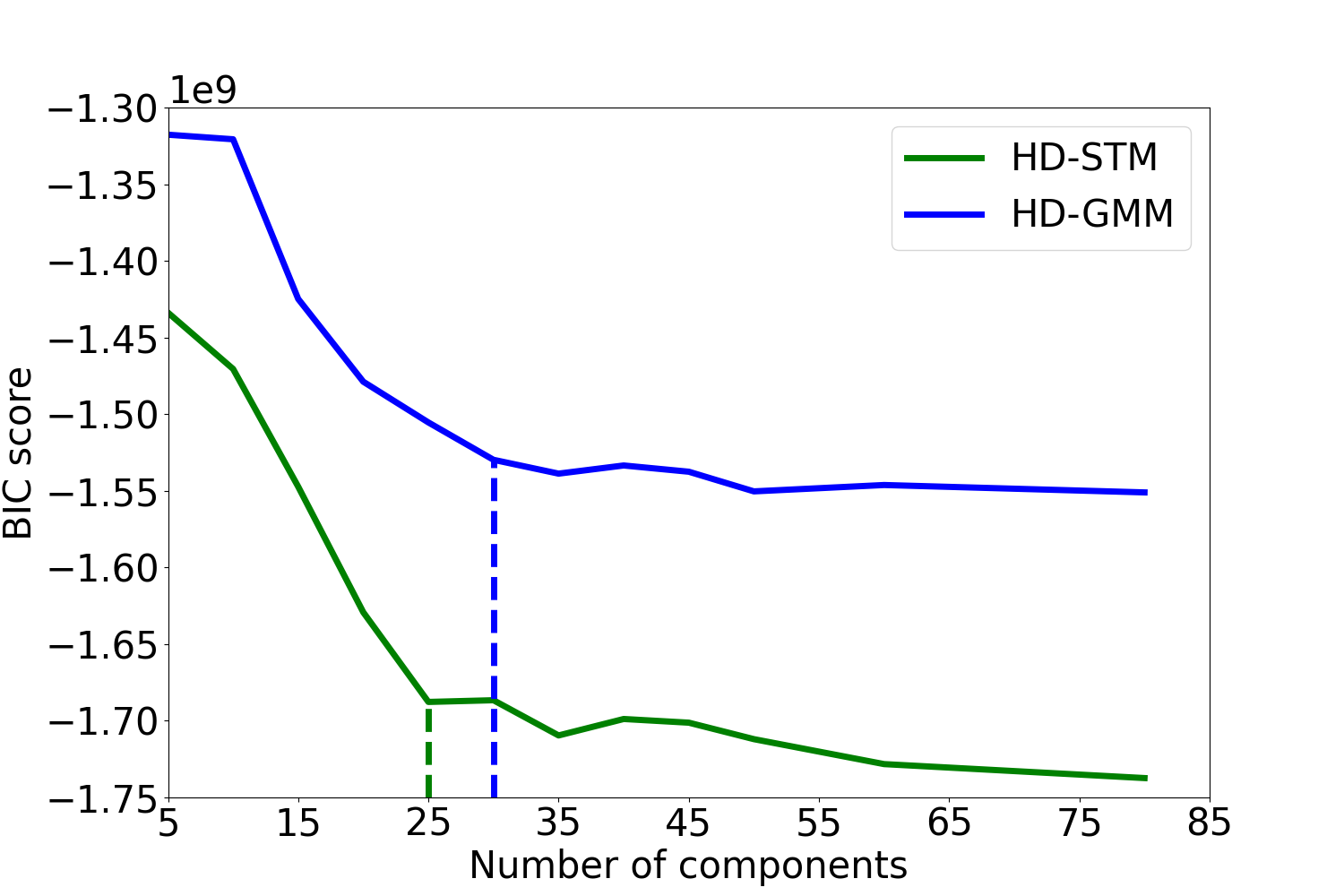}
        \includegraphics[width=0.32\textwidth]{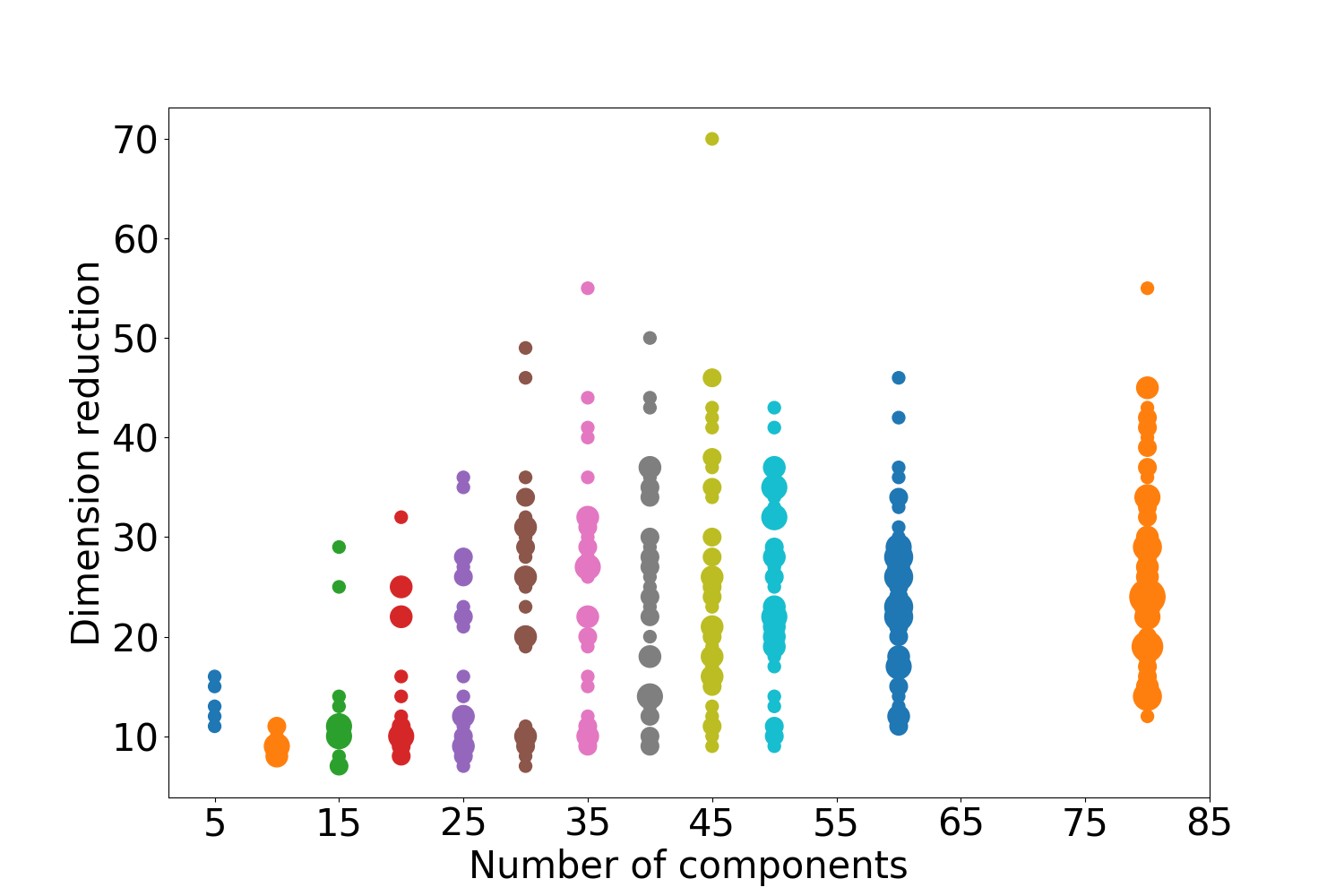}
        \includegraphics[width=0.32\textwidth]{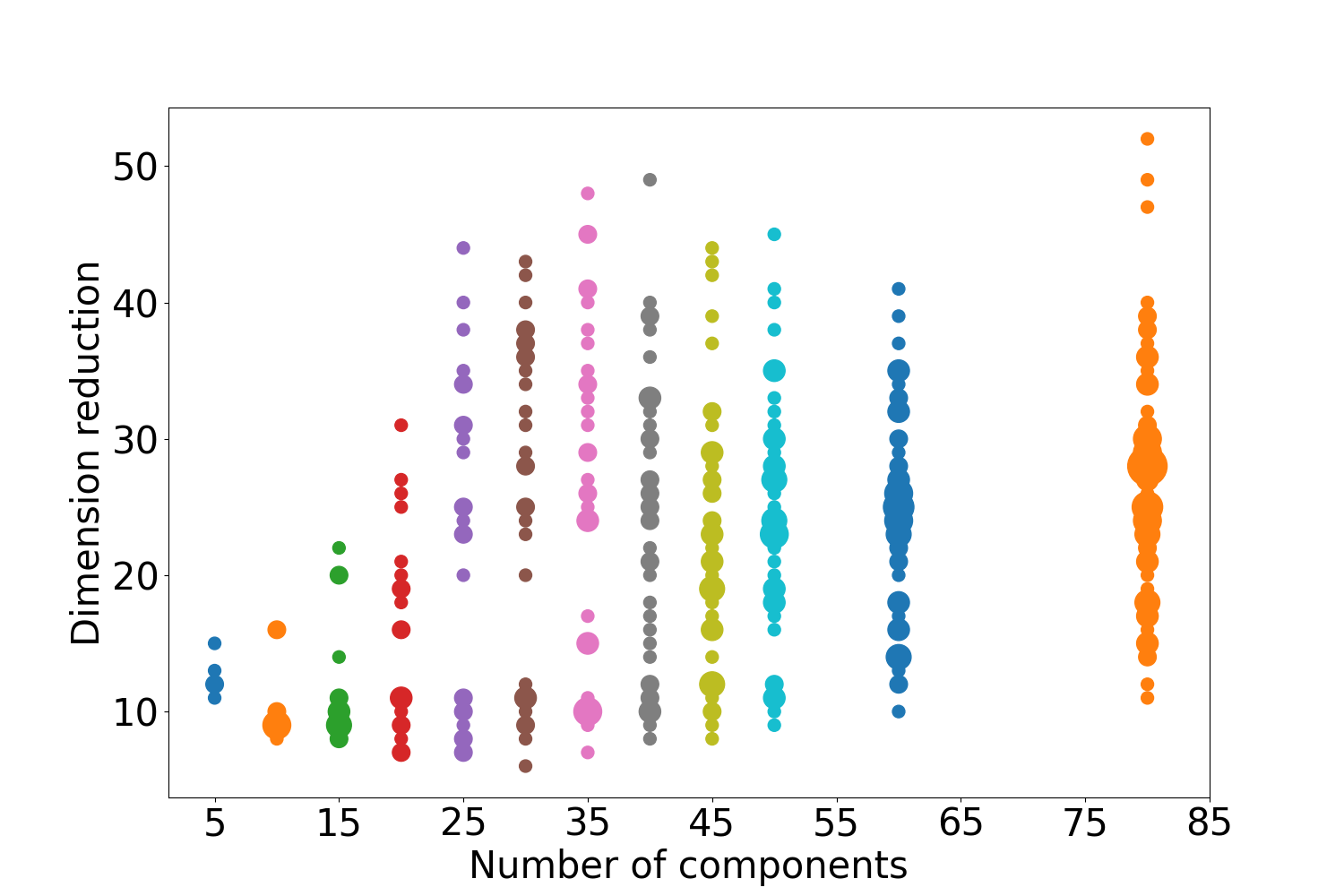}
        \caption{Model selection. BIC scores (left) and clusters dimensions, with respect to the number of components $K$, for HD-GMM (middle) and HD-STM (right) applied to signals of dimension $M=260$. Points sizes reflect the proportion of clusters with a given dimension ($y$-axis).}
        \label{fig:select}
\end{figure}

Figure \ref{fig:err_decompresion} shows that the  RMSE  decreases when $K$ increases and is smaller for HD-STM in particular for small $K$. For the selected $K$ values it represents about 2.5\% of the signal. 

\begin{figure}[h]
    \centering
    \begin{minipage}[t]{0.45\textwidth}
        \centering
        \includegraphics[width=0.95\textwidth]{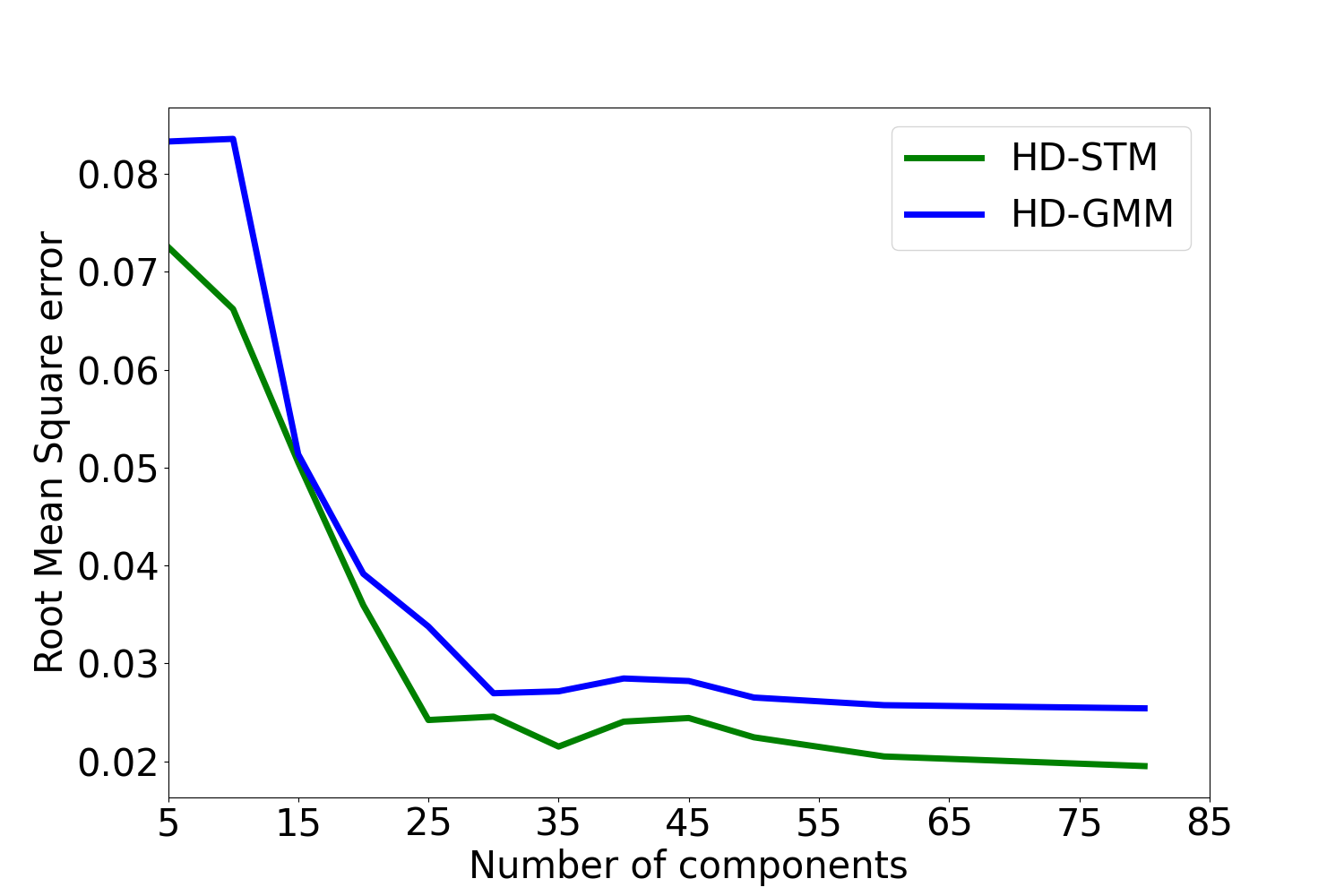}
        \caption{Dictionary reduction loss measured by RMSE over the dictionary signals, for HD-GMM and HD-STM, with respect to the number of components $K$.} 
        \label{fig:err_decompresion}
    \end{minipage}
    \hfill
    \begin{minipage}[t]{0.45\textwidth}
        \centering
        \includegraphics[width=0.80\textwidth]{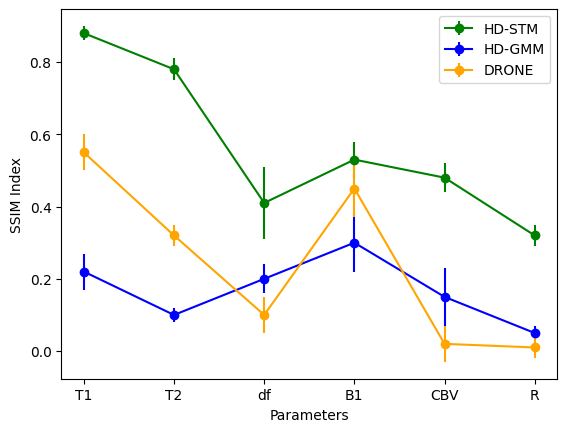}
        \caption{Parameter maps comparison. Average SSIM in [-1,1], the higher the better, and standard deviation over subjects, between HD-GMM, (resp. HD-STM, DRONE) and  full matching maps.}
        \label{fig:ssim}
    \end{minipage}
\end{figure}
Similar plots are shown in Figure 1 of Appendix Section 6 for the Laplace model HD-LM, leading to select $K=30$ and similar cluster dimension values.

\subsection{Tissue and sensitivity parameters reconstruction}
For the six subjects, we compare the parameter maps obtained with our Algorithm~\ref{alg:DC} using HD-GMM, HD-STM and with a MRF state-of-the-art neural network approach DRONE \citep{cohen2018} tailored for $T_1$ and $T_2$, which we adapted for additional parameters, to that obtained with traditional matching, referred to as \textit{full matching} (\textit{i.e.}, matching with uncompressed signals). While \textit{full matching} cannot be considered as the ground truth, it remains the reference method due to its robustness to the highly undersampled MRF acquisitions used in this study, as demonstrated in \citep{coudert2024}, despite its high computational and memory demands. Conventionally, the distance to compare two signals for matching is the Euclidean distance.

Table~\ref{tab:mae} reports for each of the 6 parameters, the distances to  \textit{full matching} parameter estimations, showing the average Mean Absolute Errors ($MAE$) across voxels for all $n_{subject}=6$ subjects and all slices. 
For another comparison that takes into account the image structure, Figure~\ref{fig:ssim} shows the reconstruction quality as measured by the structural similarity index measure (SSIM)  \citep{Wang2004} when compared to full matching, for each parameter and averaged across subjects.  The SSIM is a decimal value between -1 and 1, where 1 indicates perfect similarity, 0 indicates no similarity, and -1 indicates perfect anti-correlation.
$${MAE} = \frac{1}{n_{\text{subject}}} \sum_{s=1}^{n_{\text{subject}}} \frac{1}{\tilde{N}_s} \sum_{j=1}^{\tilde{N}_s} |\tilde{t}_{j,s} - {t}_{j,s}^{match}|,$$
$${SSIM} = \frac{1}{n_{\text{subject}}} \sum_{s=1}^{n_{\text{subject}}} {SSIM}\left ( (\tilde{t}_{j,s})_{1 \leq j \leq \tilde{N}_s}, ({t}_{j,s}^{match} ))_{1 \leq j \leq \tilde{N}_s}\right).$$

\begin{table}[h]
    \caption{Parameter maps reconstruction. Average MAE  and standard deviation, over voxels and subjects, for HD-GMM ($K=30$) HD-STM ($K=25$) and DRONE  with respect to full matching. Best values in bold.}
    \setlength{\tabcolsep}{5pt}
    \begin{tabular}{lccccccc}
        \toprule
        Parameter    & $T_1$ (ms) & $T_2$ (ms) &  $\delta f$ (Hz) & $B_1$ sensitivity ($10^{-3}$) & $\text{CBV}$ ($\%$) & $R$ ($\mu m$) \\
        \midrule
        HD-GMM &  $462 \pm 29$ & $325 \pm 17$ & $10 \pm 1$ & $92 \pm 18$ & $5 \pm 0.8$ & $2 \pm 0.03$ \\
        HD-STM & $\mathbf{78 \pm 14}$ & $\mathbf{5 \pm 2}$ & $\mathbf{6 \pm 1}$ & $\mathbf{40 \pm 5}$ & $\mathbf{1.5 \pm 0.5}$ & $\mathbf{1.6 \pm 0.06}$ \\
        DRONE  & $103 \pm 8$ & $351 \pm 10$ & $7 \pm 1$ & $60 \pm 4$ & $7 \pm 1$ & $4 \pm 0.2$ \\
        \bottomrule
    \end{tabular}
    \vspace{.2cm}
    \label{tab:mae}
\end{table}

HD-STM consistently outperforms HD-GMM and DRONE for all parameters, as evidenced by Figures~\ref{fig:ssim} and Table~\ref{tab:mae}. HD-STM achieves superior SSIM and MAE values, indicating better structural correspondence. This advantage is particularly pronounced for parameters such as $T_1$ and $T_2$, where HD-STM demonstrates robust reconstruction capabilities, whereas HD-GMM, and to a lesser extent DRONE, fail to accommodate too much undersampling noise. However, for parameters like CBV and $R$, HD-STM's performance decreases, suggesting that these parameters are inherently more challenging to model accurately. In these cases, DRONE and HD-GMM yield the lowest SSIM values and the highest MAE values, further underscoring their limitations. Overall, while HD-STM proves to be the more reliable method, the significant variability across subjects—particularly for $\delta \! f$ and CBV—and the observed SSIM drop for $R$ highlight the need for further refinements.

For another quality assessment, Table~\ref{tab:roi} presents the mean and standard deviations of $T_1$, $T_2$, $CBV$, and $R$, computed over voxels in white and grey matter ROIs delineated on $T_1$ maps obtained from \textit{full matching}. Compared to the ranges for healthy subjects reported by \cite{wansapura1999,bjornerud2010,delphin2023}, HD-STM produces values that are both more consistent with expected ranges and closer to those obtained using full matching than to those derived from HD-GMM or DRONE. In particular, with HD-GMM (resp. DRONE), $T_2$ (resp. $R$) values significantly depart from the literature reference, as also visible in Figure~\ref{fig:full_fig}.

The conclusions drawn from Figure~\ref{fig:ssim} and Tables~\ref{tab:mae}–\ref{tab:roi} are further illustrated in Figure~\ref{fig:full_fig}. This figure highlights HD-GMM and DRONE's failure to reconstruct parameters accurately, even when some estimates fall within healthy ranges reported in the literature from Table~\ref{tab:roi}. In contrast, HD-STM consistently delivers parameter estimates for $T_1$, $T_2$, $\delta \! f$, and $B_1$ that are nearly artifact-free and closely aligned with expected values. Additionally, HD-STM effectively captures primary structures and generates homogeneous maps for vascular parameters such as $CBV$ and $R$ even though minor residual \textit{shim} artifacts remain, likely resulting from significant undersampling in the acquired images. Nevertheless, the presence of these artifacts indicates that HD-STM preserves the physical characteristics of the image, rather than suppressing them.
Reconstruction results for HD-LM are shown in Appendix Section 6.1.  HD-LM performs much better than HD-GMM and DRONE but worse than HD-STM.

 \begin{figure}[h]
     \centering
     \includegraphics[width=1\linewidth]{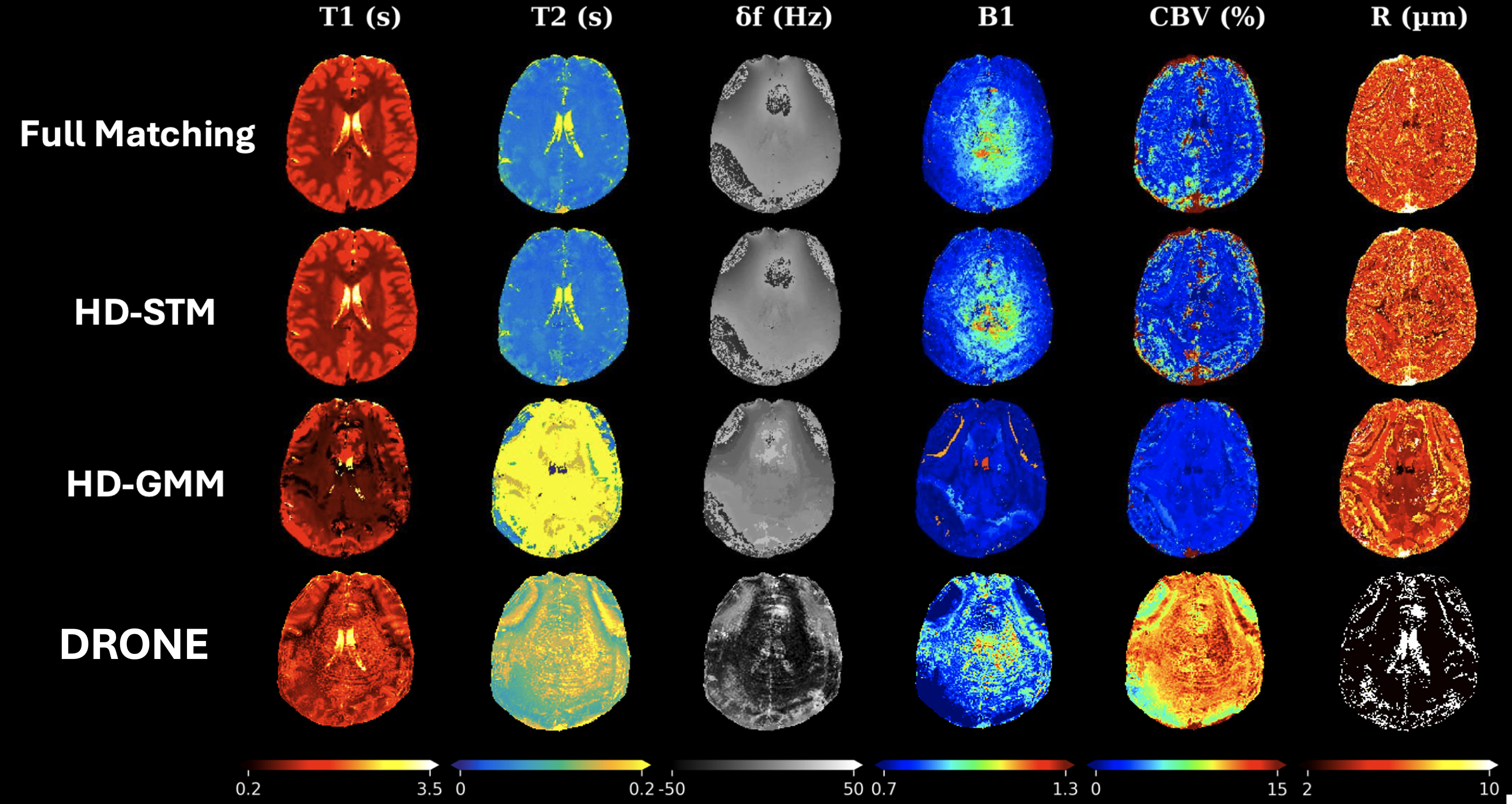}
     \caption{$T_1$, $T_2$, $\delta \! f$, $B_1$ sensitivity, $CBV$, $R$ maps (columns) with different methods (lines) - \textit{full matching}, HD-STM ($K=25$), HD-GMM ($K=30$), and DRONE. All other slices of this subject are shown in Appendix Section 8, Figure 3.}
     \label{fig:full_fig}
 \end{figure}

\begin{table}[h]
    \centering
    \caption{Mean $T_1, T_2, \text{CBV}, R$  and standard deviations,  in white (WM) and grey (GM) matter ROIs. Values significantly departing from literature (last column) are in red. 1: \cite{cohen2018}.2: \cite{WangT2etoile} }
    \begin{tabular}{llccccc}
        \toprule
        Parameter & ROI & Full matching  & HD-GMM & HD-STM & DRONE\footnote{}  & Literature\footnote{} \\
        \midrule 
        \multirow{2}[0]{*}{$T_1$(ms)} & WM &  $891 \pm 6$  & $650 \pm 23 $ & $869 \pm 19$ & $696 \pm 14$ & $690$-$1100$ \\
        & GM &  $1566 \pm 12$ &   $1650 \pm 57$&  $1631 \pm 56$ & $1660 \pm 53$ & 1286-1393\\
        \midrule
        \multirow{2}[0]{*}{$T_2$(ms)} & WM &
        $52 \pm 0.0$ &  \textcolor{red}{$400 \pm 15$} & $48 \pm 2$ & $125 \pm 5$ & 56-80  \\
        & GM & $95 \pm 4$   &  \textcolor{red}{$370 \pm 43$} & $96 \pm 12$ & $135 \pm 24$ & 78-117\\
        \midrule
        \multirow{2}[0]{*}{$\text{CBV}$($\%$)} & WM  & $5.4 \pm 0.6$   &  $2.0 \pm 0.1$ &   $4.8 \pm 0.8$ & $11.5 \pm 0.3$ & $1.7$ - $3.6$ \\
        & GM & $9.2 \pm 1.1$   & $4.2 \pm 0.8$ & $7.6 \pm 1$ & $14.6 \pm 0.6$ & $3.0$ - $8.0$  \\
        \midrule
        \multirow{2}[0]{*}{$R$ ($\mu s$)} & WM  & $5.7 \pm 0.1$   &  $5.6 \pm 0.1$ &   $5.7 \pm 0.1$ &\textcolor{red}{ $0.5 \pm 0.01$} & $6.8 \pm 0.3$ \\
        & GM & $6.2 \pm 0.1$   & $5.4 \pm 0.4$ & $6.1 \pm 0.1$ & \textcolor{red}{$0.6 \pm 0.1$} & $7.3 \pm 0.3$  \\
        \bottomrule
    \end{tabular}
    \vspace{.2cm}
    
    \label{tab:roi}
\end{table}

\subsection{Computation times}
\label{sec:times}
Training is performed on an Nvidia H100-80GB GPU, requiring 26 hours for HD-GMM, 28 hours for HD-STM, and 12 hours for DRONE. 
For HD models, these times do not include initialization, which consists of computing BIC for different $K$  on a subset of signals. The different $K$ values can be run in parallel so that initialization time mainly depends on the maximum $K$ tested and the number of sub-sampled signals. As an illustration, for $K=80$ and $10,000$ signals, initialization takes 15 minutes.
Reconstruction is executed on an Nvidia V100-32GB GPU, taking  3 minutes $\pm$ 2 seconds per slice for both HD-GMM and HD-STM, significantly faster than the 29 minutes $\pm$ 46 seconds required for \textit{full matching}. Due to memory constraints, as the full uncompressed dictionary exceeds the 32GB GPU memory capacity, an incremental matching approach was implemented. Slice reconstruction using DRONE on an Nvidia V100-32GB GPU is notably rapid, completing in under 1 minute. This highlights the efficiency of neural network-based regression methods. However, DRONE exhibits limitations in accurately estimating vascular parameters (Figure \ref{fig:full_fig}). 
%
In more realistic scenarios, calculations are performed locally by a medical practitioner on a CPU. On an Apple M2 Pro CPU, \textit{full matching} takes days and is prone to many memory issues that need to be solved, while the HD-GMM and HD-STM variants require only $4.5$ hours. 
This is a strength of HD models over complex neural network approaches, which cannot be easily adapted to CPU and whose running times would be prohibitive.

\section{Conclusion}

In this work, we combine robust latent variable representations, clustering and incremental learning to propose a tractable and accurate way to represent and handle large volumes of potentially heterogeneous high-dimensional data. To our knowledge, this combination and its use as a data compression strategy is novel.  
The clustering structure of HD-MED allows to address data heterogeneity, for a greater dimensionality reduction without increasing information loss. Incremental learning allows to handle large volumes, resulting in  significant reduction in both computational costs and information losses. 
In terms of implementation, the procedure is flexible and easy to interpret. It depends on two main hyperparameters, the number of clusters $K$ and the vector of reduced dimensions $\db$, that  can be set using  conventional model selection criteria or chosen by the user  to meet resource constraints, {\it e.g.}  increasing $K$ to reduce the size of clusters or decreasing $\db$ to increase compression.

As an illustration, we focus on a Magnetic Resonance Fingerprinting (MRF) application. The proposed method drastically reduces the computational time required on standard hardware, such as CPUs, which are commonly used in clinical environments. Beyond this crucial processing time gain, dimensionality reduction has also an interesting impact on patient data acquisition time, as it helps mitigate the effect of noise in fast-acquired \textit{in vivo} signals. This reduction is thus not only a technical achievement but also a significant step forward in making advanced MRF techniques more accessible and usable in daily clinical practice. Currently, the sequence used to acquire these images is in the pre-clinical testing phase on pathological animal, including those with tumors and stroke. With respect to cross-vendor data, since MRF provides quantitative measurements, there are fewer harmonization issues between scanners from different vendors. However, a potential limitation is the noise, which can differ between scanners; our simulated dictionary originally contains denoised signals, which we artificially noised with different values of the signal-to-noise ratio (SNR), but we did not adjust the noise to exactly match that of the acquisition data, also suggesting a potentially good generalization between acquisitions from different scanners.

In addition, although illustrated with a simple matching procedure, which is the current reference in the target MRF application, our procedure can be coupled with other simulation-based inference (SBI) approaches \citep{Cranmer}, such as approximate Bayesian computation (ABC) or other neural techniques \citep{barrier2024}. 
Like MRF, simulation-based inference  has to face two opposite requirements, which are the need for large and high-dimensional simulated data sets to accurately capture information on the physics under study and the issue of handling such large volumes due to computational resource constraints. 
ABC uses distances between observations and simulations,  and inference is based on these distances. A reduced data representation can  typically be coupled with ABC techniques proposing automatic summary statistics selection such as in \cite{forbes2022,fearnhead2012}. It can also be used with SBI neural approaches, {\it e.g.} \cite{haggstrom2024}, when coupled with a regression model for high-dimensional data, see Algorithm 1 in supplementary material. Compared to matching and standard ABC, regression-based procedures are more amortized, which further reduces the reconstruction time. 
In terms of MRF and other medical imaging applications, future work will thus involve combining HD-MED with a regression model or a neural network, such as  in \cite{Deleforge2015} or  \cite{golbabaee2019}. 
More generally, the possibility to incrementally learn HD-MED models, from high-dimensional heterogeneous data in large volumes, opens the way to a broad range of applications. 
A first type of applications relates to 
anomaly detection using a likelihood score, see  \cite{Lagogiannis2024} for a recent review. 
Anomaly detection can be described as an outlier detection
problem, where the aim is to discriminate between in- and
out-of-distribution samples of a normative distribution learned from a dataset of normal samples. These normal samples can typically be high-dimensional feature vectors summarizing a possibly heterogeneous set of subjects.
Second, inverse problems in an SBI context provide numerous applications other than MRF, such as remote sensing tasks in planetary science, extending, for instance, the work by \cite{Kugler2022}.
 Third, HD-MED  could also simply be used as an alternative to standard PCA to achieve a better low-dimensional representation of datasets of interest, or simply as a clustering technique \citep{bouveyron2014}.

Lastly, none of the mentioned methods makes use of spatial information, to perform map reconstruction. The idea would be to account for voxel proximity to improve or accelerate prediction.
A natural extension of mixture-based methods is to account for spatial information by adding a Markov dependence on the clusters, as done  by \cite{DeleforgeFBH15}.

\begin{acks}[Acknowledgments]
The authors would like to thank the referees, an Associate
Editor and the Editor for their constructive comments that improved the
quality of this paper.
\end{acks}

\begin{funding}[Financial disclosure]
 GO was supported by the AURA region and was granted access to the HPC resources of IDRIS under the allocation 2022-AD011013867 made by GENCI.
\noindent The MRI facility IRMaGe is partly funded by the program “Investissement d’avenir” run by the French National Research Agency, grant “Infrastructure d’avenir en Biologie et Santé” [ANR-11-INBS-0006]. The project is supported by the French National Research Agency [ANR-20-CE19-0030 MRFUSE].
\end{funding}

\begin{supplement}
Supplementary material includes an additional discussion on the pratical benefits of the online EM algorithm. Details on its application to finite mixture models, and in particular to mixtures of Student and Laplace distributions in a high-dimensional settting are also provided. More illustrations on the fingerprinting application can also be found.
The code to run the proposed algorithm and reproduce key results is also accessible on GitHub (\href{https://github.com/geoffroyO/onlineHDEM}{https://github.com/geoffroyO/onlineHDEM}). The repository includes a notebook, \textit{tutorial\_hd\_models.ipynb}, which provides a comprehensive tutorial on how to run the different models with simulated data. However, for regulatory and patient privacy reasons, the data used in this study cannot be made publicly available.
\end{supplement}

\bibliographystyle{imsart-nameyear}
\bibliography{bibliography}

\end{document}


\begin{frontmatter}
\title{Supplementary material\\Scalable magnetic resonance fingerprinting: Incremental inference of high dimensional elliptical mixtures  from  large data volumes}
\runtitle{Scalable magnetic resonance fingerprinting}
\end{frontmatter}

%

\section{Practical benefit of online EM vs batch EM}

Dictionary-based techniques aim at being as amortized as possible, meaning that most of the computational burden is performed offline, ahead of the task of interest, which can then be performed straightforwardly for each new patient. Thus, in principle, the standard batch EM could be executed once, using powerful cloud computing resources, without resorting to a more frugal incremental EM. However, there are several reasons why this may not be the case in practice:
\begin{itemize}
    \item For large  dictionaries, the resources needed for a batch EM are huge, so that not all physicians / hospitals / researchers can afford the  resource and time costs. A useful dictionary exceeds in size several terabytes in memory, which requires an excessive amount of RAM to process efficiently. On typical computer clusters, it is hard to find a single node with sufficient RAM, necessitating the use of multiple nodes and thus, resorting to  implementing  a distributed EM algorithm. The incremental EM setting makes the implementation simpler. 
    \item In addition, building a dictionary can be an experimental design task. For practitioners, the batch EM cost can be detrimental and prevent the necessary repetitive trials to design an appropriate dictionary for a given task. Note that the dictionary considered in the paper is already 400 million  signals for a size of 13 terabytes, which, for 6 parameters,  represents  approximately a grid made with  only 27 values for each parameter.
Although the hope is that a given dictionary be used for several patients, depending on the study, other parameter values may need to be included or changed, so that the dictionary may have to be changed or be completed with more signals. In the latter case, the incremental EM allows either to re-learn a dictionary representation from scratch, without excessive cost, or to pursue the previous learning with the newly added dictionary elements.
    In other words, incremental EM enables to refine a dictionary on the fly by adding more simulated signals—such as transitioning to a finer grid—very easily. We do not need to re-run the complete batch EM algorithm but can simply pass mini-batches of new simulated signals incrementally to the model. 
    \item Finally, beyond computational costs,  incremental/stochastic methods often enable better generalization. This is particularly the case when Stochastic Gradient Descent (SGD) is compared to Gradient Descent (GD) in deep learning \citep{Amir2021}.
    In the case of the EM algorithm, experimental results also show that this can hold. For instance, Figure 1 of \cite{gilman2023bis}
    demonstrates that when only 50\% of the dataset is observed, the incremental EM achieves better results in terms of log-likelihood compared to the batch version of EM (HePPCAT).
\end{itemize}
To summarize, compared to standard EM, the proposed incremental EM allows 1) much more flexible experimental design and investigations, and 2) better learning performance. So doing it broadens the applicability and accuracy of MRF in general, and in particular beyond the settings studied in the paper.

\section{Online EM vs Stochastic Gradient Descent (SGD)}

SGD can be used to estimate mixtures but this is not yet very popular. Both EM and GD approaches have strengths. 
Gradient approaches do not need to compute E and M steps and can thus easily be adapted from one model to another. However, 
although automatic differentiation is a powerful tool, its use when there are constraints on the parameters is not straightforward and requires reparametrization and sometimes over-reparametrization, typically when parameters are proportions that sum to 1, covariance matrices or orthogonal matrices. See, for instance, an elegant reparametrization used by \cite{Hosseini2015} that handle covariance matrices by switching to a Riemannian manifold. Handling orthogonality constraints, like for matrix $\Db$ is likely to be more involved. 
In contrast, incremental EM, which works on sufficient statistics, does not need smoothness assumptions and gradient computations, see also \cite{Dieuleveut2023} for a recent analysis of stochastic non-gradient algorithms.

In addition,  GD/SGD is extremely sensitive to initialization, even more than EM for which initialization can also be important.
EM avoids certain poor local minima that GD/SGD might fall into, as shown in analyses of two-component mixtures \citep{Zhang2020}.
EM convergence is generally faster. EM can be seen as a preconditioned version of GD, see {\it e.g.} \cite{XuJordan1996}.

Another feature is that  SGD is more memory - and time - intensive than incremental EM, primarily because automatic differentiation tools like autodiff construct and store computational graphs to compute gradients, increasing overhead for large-scale, high-dimensional problems. In contrast, incremental EM can often leverage closed-form updates for the E- and M-steps,
avoiding the need for such graphs.

EM-based methods, including incremental versions, are then frequently more robust and efficient for mixture estimation. 
However, \cite{Gepperth2021}
 compare incremental EM with SGD. To allow SGD to escape local optima, they propose an annealing procedure. Their conclusion is that both SGD  and incremental EM perform similarly but SGD can outperform incremental EM in high dimensions. In our work, high dimensionality is handled via our HD-MED model and we chose to benefit from the other strengths of EM.

\section{Online EM for finite mixture models}
This section completes Section 6.1 of the main paper. We demonstrate that the online EM (OEM) algorithm for a finite mixture model can be derived from the OEM algorithm for a single component. We denote the mixture parameters as $\Thetab_\mathcal{M}$, which include the pairs $(\pi_k, \Thetab_k)$ for $k = 1, \dots, K$, where $\Thetab_k$ are the component-wise parameters of the finite mixture model. Furthermore, we consider the pair $(\Yv^T, Z)$, where $Z = 1:K$ is a latent variable such that $\mathbb{P}(Z = k) = \pi_k$.

We assume that the complete-data likelihood for each component can be expressed in the exponential family form, as given by Equation (12) in the main paper, for each $f_c(\yv, z; \Thetab_k)$. Following \cite{NguyenForbes2021}, the complete-data likelihood for the mixture model $f_c(\yv, z, \Thetab_\mathcal{M})$ can be written as:
\begin{align}
    f_c(\yv, z, \Thetab_\mathcal{M}) &= h(\yv, z) \exp \left[ \sum \limits_{\xi =1}^K \left[\log \pi_{\xi} + [\sb(\yv, z)]^T \phib(\Thetab_{\xi}) - \psi(\Thetab_{\xi})\right]\right] \notag \\
    &= h_\mathcal{M}(\yv, z) \exp \left[[\sb_\mathcal{M}(\yv, z)]^T \phib_\mathcal{M}(\Thetab_\mathcal{M}) - \psi_\mathcal{M}(\Thetab_\mathcal{M})\right],
\end{align}

where $h_\mathcal{M} = h$, $\psi_\mathcal{M} = \mathbf{0}$,

\begin{equation}
    \label{eq:formulas_fm_exp}
    \begin{aligned}
    \sb_\mathcal{M}(\yv, z) &= \begin{bmatrix}
           \mathbf{1}_{z=1} \\
           \mathbf{1}_{z=1} \sb(\yv, z) \\
           \vdots \\
           \mathbf{1}_{z=K} \\
           \mathbf{1}_{z=K} \sv(\yv, z)
         \end{bmatrix}, \text{ and } \hspace{.5cm}
    \phib_\mathcal{M}(\thetab_\mathcal{M}) = \begin{bmatrix}
          \log \pi_1 - \psi(\Thetab_1) \\
          \phib(\Thetab_1) \\
          \vdots \\
          \log \pi_K - \psi(\Thetab_K) \\
          \phib(\Thetab_K)
         \end{bmatrix}.
    \end{aligned}
\end{equation}

We can conveniently write the vector $\sb_\mathcal{M}$ that matches the dimension of $\sb_\mathcal{M}(\yv, z)$ as
\begin{equation}
    \sb_\mathcal{M}^T = \left(s_{01},\sb_{\mathcal{M}1}^T, \dots, s_{0K},\sb_{\mathcal{M}K}^T \right),
\end{equation}
where $\sb_{\mathcal{M}z}$ corresponds to one component of the exponential family form of the mixture model distribution.
Via Equation~\eqref{eq:formulas_fm_exp}, the objective function $Q_{\mathcal{M}}$ for the mixture complete-data likelihood as defined in Equation (14) of the main paper can be written as
\begin{equation}
    Q_{\mathcal{M}}(\sb_\mathcal{M}, \Thetab_\mathcal{M}) = \sb_\mathcal{M}^T \phib_\mathcal{M}(\Thetab_\mathcal{M}) = \sum \limits_{z=1}^K s_{0z} \left(\log \pi_z - \psi(\Thetab_z)\right) + \sb_{\mathcal{M}z}^T \phib(\Thetab_z).
\end{equation}
Whatever the form of the component \textit{p.d.f}, the maximisation with respect to $\pi_z$ yields
\begin{equation}
    \label{eq:pi_mix}
    \Bar{\pi}_z(\sb_{\mathcal{M}}) = \frac{s_{0z}}{\sum \limits_{\xi = 1}^K s_{0\xi}}.
\end{equation}
Then, for each $z \in \{1 \dots K\}$,
\begin{align}
    \frac{\partial Q_\mathcal{M}}{\partial \Thetab_z}\left(\sb_\mathcal{M}, \thetab_\mathcal{M}\right) &= -s_{0z}\frac{\partial \psi}{\partial \Thetab_z} (\Thetab_z) + \Jb_{\phib}(\Thetab_z)\sb_{\mathcal{M}z} \notag \\
    &= s_{0z} \left(\Jb_\phib (\Thetab_z) \left[\frac{\sb_{\mathcal{M}z}}{s_{0z}}\right] - \frac{\partial \psi}{\partial \Thetab_z}\right) \notag \\
    &= s_{0z} \frac{\partial Q}{\partial \Thetab_z} \left(\left[\frac{\sb_{\mathcal{M}z}}{s_{0z}}\right], \Thetab_z  \right),
\end{align}
where $Q$ is the objective function of form Equation (14) from main paper corresponding to a single component \textit{p.d.f}. Since $s_{0z}>0$, for all $z \in \{1 \dots K\}$, it follows that the maximization of $Q_\mathcal{M}$ can be conducted by solving
\begin{equation}
    \frac{\partial Q}{\partial \Thetab_z} \left(\left[\frac{s_{\mathcal{M}z}}{s_{0z}}\right], \Thetab_z \right) = \mathbf{0},
\end{equation}
with respect to $\Thetab_z$ for each z. Therefore, it is enough to show that for each component $z$ of the finite mixture model that there exists a root of the equation above denoted as $\Bar{\Thetab}_z(\sb_{\mathcal{M}z}/s_{0z})$ in order to find a solution for the maximizer of the mixture objective $Q_\mathcal{M}$. Then we can set

\begin{equation}
    \Bar{\Thetab}_\mathcal{M}(\sb_\mathcal{M}) = \begin{bmatrix}
        \Bar{\pi}_1 \left(s_\mathcal{M}\right) \\
        \Bar{\Thetab}_1\left(\sb_{\mathcal{M}1} / s_{01}\right) \\
        \vdots \\
        \Bar{\pi}_K \left(s_\mathcal{M}\right) \\
        \Bar{\Thetab}_K\left(\sb_{\mathcal{M}K} / s_{0K}\right)
    \end{bmatrix},
\end{equation}
where the form of $\Bar{\pi}_K \left(s_\mathcal{M}\right)$ is given in \eqref{eq:pi_mix}.

To complete the OEM algorithm, we need to specify the quantity $\Bar{\sb}_{\mathcal{M}}(\yv; \Thetab_\mathcal{M}) = \mathbb{E} \left[\sb_{\mathcal{M}}\left(\Yv, Z\right) | \Yv = \yv;\Thetab_\mathcal{M}\right]$ as defined in Equation (13) of the main paper. Using the definition of $\sb_\mathcal{M}(\yv, z)$ in Equation~\eqref{eq:formulas_fm_exp}, we compute,

\begin{equation}
    r_{k} = \mathbb{E} \left[\mathbf{1}_{\{Z=k\}}| \Yv = \yv; \Thetab_\mathcal{M}\right] = \frac{\pi_k f_k(\yv)}{\sum \limits_{\xi=1}^K \pi_k f_k(\yv)},
\end{equation}
and
\begin{equation}
    \mathbb{E} \left[\mathbf{1}_{\{Z=k\}}\sb(\Yv, Z)| \Yv = \yv; \Thetab_\mathcal{M}\right] = r_k \mathbb{E} \left[\sb(\Yv)| \Yv = \yv , Z=k; \Thetab_\mathcal{M}\right ],
\end{equation}

where $f_k$ is defined in Equation (4) of the main paper as the conditional pdf of the $k^{th}$ mixture component.

Thus, in the finite mixture framework, deriving the OEM algorithm for a single component suffices to deduce the OEM algorithm for the entire mixture. In the next section, we derive the OEM algorithm for a single component of HD-MED. For simplicity, we omit further mention of the form of $h$ in the exponential family representation, as it does not affect the OEM algorithm derivation.

\section{HD-ED exponential family form: Proof of Proposition 6.1}

\begin{proof}
Recall that $\widetilde{\Db}$ denotes the $d$ first columns of $\Db$ completed with zero columns and $\overline{\Db} = \Db- \widetilde{\Db}$. The following multivariate normal distribution can be decomposed as follows,
\begin{gather}
\label{eq:normdecomp}
    \mathcal{N}_M\!\left(\yv; \mub, \frac{\Db \Ab \Db^T}{w}\right) 
    = \prod \limits_{m=1}^{d} \mathcal{N}_1\!\left([\widetilde{\Db}^T(\yv-\mub)]_m; 0, \frac{a_{m}}{w}\right) \prod \limits_{m=d+1}^{M} \mathcal{N}_1\!\left([\overline{\Db}^T(\yv-\mub)]_m; 0, \frac{b}{w}\right), 
    \notag
\end{gather}
where $[\cdot]_m$ denotes the $m^{\text{th}}$ element of the input vector. For $m \leq d$ it comes, 
\begin{equation*}
    \label{eq:vectrick}
    [\widetilde{\Db}^T(\yv-\mub)]^2_m = vec(\db_{m}\db_{m}^T)^Tvec(\yv \yv^T) + vec(\db_{m}\db_{m}^T)^Tvec(\mub \mub^T) - 2\mub^T\db_{m}\db_{m}^T\yv.
\end{equation*}
And for $d < m \leq M$ we have:
\begin{align}
    \label{eq:secpart}
    & \prod \limits_{m=d+1}^{M} \mathcal{N}_1\left([\overline{\Db}^T(\yv-\mub)]_m; 0, \frac{b}{w}\right) \notag 
    &= \frac{1}{(2 \pi b)^{\frac{M-d}{2}}} \exp\left(-\frac{w}{2b}\sum\limits_{m=d+1}^{M} \overline{\Db}^T(\yv-\mub)]^2_m \right).
\end{align}
It is then  possible to remove the dependence  on $\overline{\Db}$ by noting that
\begin{align*}
    \sum \limits_{m=d+1}^{M}[\overline{\Db}^T(\yv-\mub)]^2_m &= ||\yv - P(\yv)||^2 \\ 
    &= \yv^T\yv + \mub^T\mub - 2\yv^T\mub - \sum \limits_{m=1}^{d}[\widetilde{\Db}^T(\yv-\mub)]^2_m,
\end{align*}
using that $P(\yv) = \widetilde{\Db}\widetilde{\Db}^T(\yv-\mub) + \mub $.
Rewriting expressions with the $\operatorname{vec}$ operator and  combining both parts where $m \leq d$ and $d \leq m \leq M$ gives the exponential family form for the complete data likelihood (18) in main paper.
\end{proof}

\section{Gamma-distributed mixing variable}
\label{ssec:gamma}

In Section 6 of the main paper we proposed an online E-step with incomplete calculations that depend on the exponential family form of mixing variable $W$. In this section we propose to specify these calculations when $W$ follows a  Gamma distribution $f_\thetab(w) = \mathcal{G}\left(w; \alpha, \beta\right)$ with the two parameters $\thetab=(\alpha,\beta)$ in $\Rset^{+2}$. More specifically, we need to compute the expectation (19), $\mathbb{E}[W | \Yv ; \Thetab]$, from the main paper, and $\mathbb{E}[s_{w}(W)| \Yv ; \Thetab]$. Then we need to derive the root equations for $\alpha$ and $\beta$. This Gamma-distributed mixing $W$  encompasses several distributions as indicated in Table~\ref{tab:gamma_exemple}.

The Gamma distribution has density $\mathcal{G}\left(w; \alpha, \beta\right) = w^{\alpha-1} \frac{\beta^\alpha e^{-\beta w}}{\Gamma(\alpha)}$ with $\Gamma$ the Gamma function. This distribution is easily put in the exponential family form with
\begin{equation}
    \label{eq:exp_form}
    \begin{aligned}
    \sb_{w}(w) &= \begin{bmatrix}
           w \\
           \log w \\
         \end{bmatrix}, \hspace{.5cm}
    \phib_{w}(\alpha, \beta) = \begin{bmatrix}
          -\beta \\
          \alpha \\
         \end{bmatrix}, \\
         \psi_{w}(\alpha, \beta) &= \log (\Gamma(\alpha)) - \alpha \log (\beta).
    \end{aligned}
\end{equation}

The next proposition gives the necessary formulas to derive the online E-step for an ED with a Gamma-distributed mixing variable.
\begin{proposition}[Online E-step with Gamma mixing variable]
    \begin{align}
        \mathbb{E}\left[W | \Yv=\yv\right] &= \frac{2 \alpha + M}{u + 2 \beta} \\
        \mathbb{E}\left[\operatorname{log}(W) | \Yv=\yv\right] &=  \Psi^{(0)}\left(\frac{2 \alpha + M}{2}\right) - \operatorname{log}\left(\frac{u + 2 \beta}{2}\right),
    \end{align}
with $u= (\yv-\mub)^T\Db\Ab^{-1}\Db^T(\yv-\mub)$ and $\Psi^{(0)}$ the diGamma function defined as the derivative of $\log \Gamma(\cdot)$.
\end{proposition}
\begin{proof}
    We have $(\Yv | W=w)\sim \mathcal{N}_M\left(\mub, \Sigmab / w\right)$ and $W \sim \mathcal{G}\left(\alpha, \beta \right)$, then
    \begin{align*}
        p(w | \Yv= \yv) &\propto w^{\alpha + \frac{M}{2}-1} \exp \left(-w\left(\frac{u}{2} + \beta \right)\right) \\
        &\propto \mathcal{G}\left(w; \alpha + \frac{M}{2}, \frac{u}{2} + \beta \right),
    \end{align*}
which provides the values for both required expectations.
\end{proof}

For the online maximization step, we need to find the solution of the root equation (26) in the main paper. By conveniently writing  $\sb_{w}^T =  [s_{1w} \text{ } s_{2w}]$ and differentiating with respect to $\alpha$ and $\beta$, it comes the  following system of equations,
\begin{align}
    s_{2w} - \Psi^{(0)}(\alpha) + \log \beta &= 0 \label{eq:alpha_eq}\\
    -s_{1w} + \frac{\alpha}{\beta} &= 0 \; . 
\end{align}
Then substituting  $\alpha / s_{1w}$ to $\beta$ in (\ref{eq:alpha_eq}) and under the condition that $s_{1w} \neq 0$, (\ref{eq:alpha_eq}) leads to an equation with $\alpha$ only.

\subsection{Mixture of HD Student-t distributions}
As recalled in  Table~\ref{tab:gamma_exemple}, the Student distribution is an ED with a Gamma-distributed mixing variable with $\alpha = \beta = \frac{\nu}{2}$ where $\nu$ is usually referred to as the degrees of freedom. Previous formulas apply and the root equation for $\nu$ simplifies into,
\begin{equation*}
    s_{2w} - s_{1w} - \Psi^{(0)} \left( \frac{\nu}{2} \right) +  \log \left(\frac{\nu}{2}\right) + 1 = 0.
\end{equation*}

\begin{table}[htbp]
    \centering
    \caption{Some EDs with Gamma mixing variable.}
    \begin{tabular}{cc}
        \hline 
        ED & Mixing variable \\
        \hline 
        Student  & $W \sim \mathcal{G}\left(\frac{\nu}{2}, \frac{\nu}{2}\right)$ \\
        Pearson Type VII & $W \sim \mathcal{G}\left(\alpha - \frac{p}{2}, \frac{1}{2}\right)$ \\
        Generalized Student & $W \sim \mathcal{G}\left(\frac{\nu}{2}, \frac{\alpha}{2}\right)$\\
        Cauchy & $W \sim \mathcal{G}\left(\frac{1}{2}, \frac{1}{2}\right)$\\
        \hline
    \end{tabular}
    \label{tab:gamma_exemple}
\end{table}

\section{Inverse Exponential mixing variable and Laplace distribution}

Another example of widely used distribution is the symmetric multivariate Laplace distribution whose pdf in dimension $M$ is
$$p(\yv; \mub,\Sigmab) = \frac{2}{(2\pi)^{M/2}  (\det\Sigmab)^{1/2}} \left(\sqrt{\frac{u}{2}}\right)^{1-\frac{M}{2}} K_{1-\frac{M}{2}}(\sqrt{2u}) $$
where $\mub$ and $\Sigmab$ are location and scale parameters, $u = (\yv - \mub)^T \Db \Ab^{-1} \Db^T (\yv - \mub)$ is the Mahalanobis distance and  $K_{\cdot}$ is the modified Bessel function of the second kind, see {\it e.g.} Appendix of \cite{Kotz2001}.
This is equivalent to the definition used by \cite{Eltoft2006} by setting  the parameter $\lambda$ therein to 1 and removing the constraint $\det\Sigmab=1$, which for non fixed $\lambda$ is necessary  to guarantee identifiability.

This multivariate distribution, can be written as a GSM with a mixing variable $W$ that follows an inverse exponential distribution, {\it i.e.} $W^{-1}=Z \sim {\cal E}(1)$.
In our notation, the mixing variable $W$ is the inverse of the variable $Z$ used by \cite{Eltoft2006}. It follows that the quantity required for our online EM, namely $\mathbb{E}\left[ W| \Yv=\yv\right]$ can be directly deduced from formula (20) in \citep{Eltoft2006},
\begin{equation}
    \mathbb{E}\left[ W| \Yv=\yv\right] = \sqrt{\frac{2}{u}} \frac{K_{-\frac{M}{2}}\left(\sqrt{2 u}\right)}{K_{-\frac{M}{2} + 1}\left(\sqrt{2 u}\right)}. \label{m1}
\end{equation}
We can then update the sufficient statistics following equation (15) in the main text. 
Applying Proposition 6.3,  estimations for $\mub$, $\Ab$ and $\tilde{\Db}^*$ follow straightforwardly. 

Note that the above moment can also be obtained by starting from the expression of the multivariate Laplace distribution as a GSM
\begin{align*}
    p(\yv;\mub,\Sigmab) &= \int \limits_{\mathbb{R}^+}\mathcal{N}_M \left( \yv; \mub;z \Sigmab \right)  e^{-z}dz \\
    &= \int \limits_{\mathbb{R}^+}\mathcal{N}_M \left( \yv; \mub; \frac{\Sigmab}{w} \right) w^{-2}  e^{-\frac{1}{w}}dw, 
\end{align*}
and deriving that  
\begin{align*}
    p(w | \Yv = \yv) \propto w^{\frac{M}{2} - 2} e^{-\frac{1}{2}w u -\frac{1}{w}}
    \propto \operatorname{GIG}\left(w; \frac{M}{2} - 1, u, 2\right),
\end{align*}
where  
 $\operatorname{GIG}$ is the pdf of the Generalized Inverse Gaussian distribution, using the parametrization of \cite{Jorgensen1982}.
The moments of the GIG are known and given, for $r \in \Rset$, by \citep{Jorgensen1982},
\begin{equation*}
    \mathbb{E}\left[ W^r| \Yv=\yv\right] = \left(\sqrt{\frac{2}{ u}}\right)^r \frac{K_{\frac{M}{2} -1+r}\left(\sqrt{2 u}\right)}{K_{\frac{M}{2} - 1}\left(\sqrt{2 u}\right)}.
\end{equation*}
For $r=1$ it comes  (\ref{m1}), using the fact that the modified Bessel functions satisfy $K_p = K_{-p}$. 
In practice we use \cite{takekawa2022fast} to make efficient computation of the ratio $\frac{K_{\cdot + 1} (\cdot)}{K_{\cdot} (\cdot)}$ in Python-Jax.

\subsection{MRF reconstruction with High-Dimensional Mixture of Multivariate Laplace distributions (HD-LM)}
We run the experiment to reconstruct the MRF signals using the same data and procedure, from Section 7 of the main paper. Training the model multiple times with different numbers of components and computing  BIC gives an optimal number of classes of $K=30$. Figure~\ref{fig:select_hdlm} shows the BIC values and the different dimensions of the sub-spaces, with respect to the number of components. Table~\ref{tab:mae_ssim_hdlm} shows the average MAE and SSIM over the 6 subjects, when compared to full matching. 
The results show clearly that the Laplace mixture model HD-LM performs much better than HD-GMM and DRONE but less than HD-STM.
\begin{figure}[h]
        \includegraphics[width=0.4\textwidth]{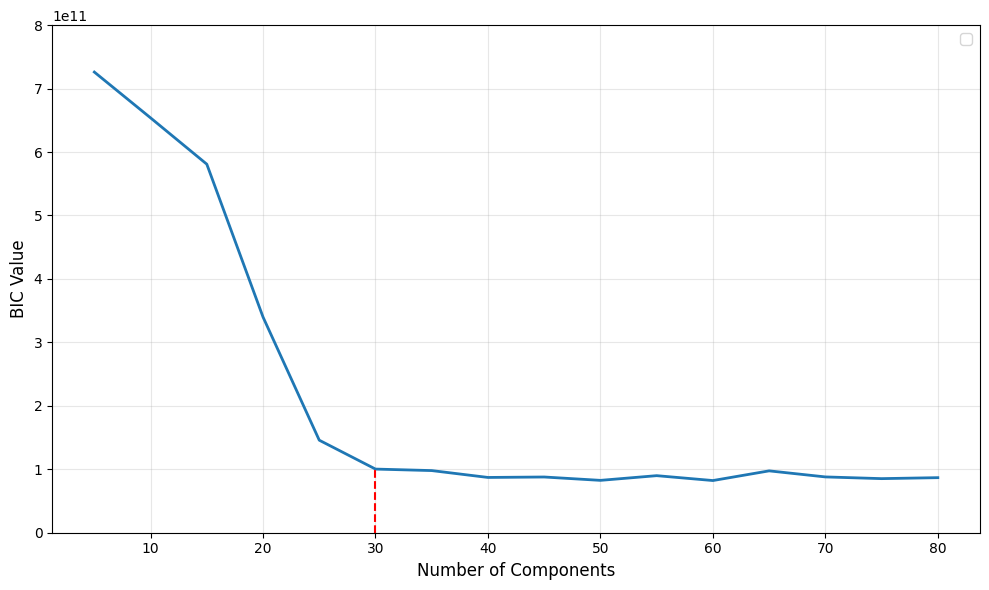}
        \includegraphics[width=0.25\textwidth]{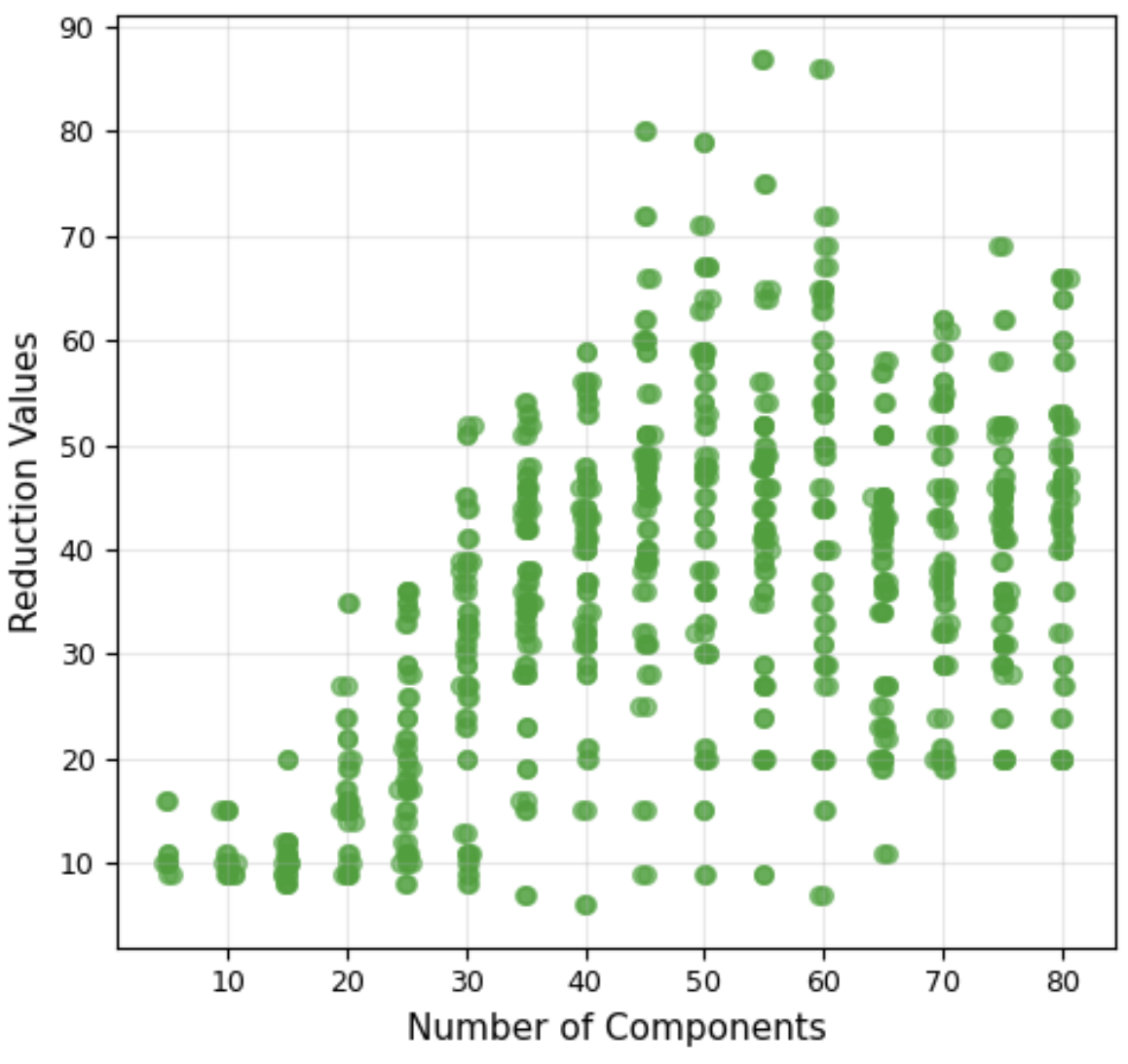}
        \caption{Model selection. BIC scores (left) and clusters dimensions (right), with respect to the number of components $K$, for HD-LM  applied to signals of dimension $M=260$.}
        \label{fig:select_hdlm}
\end{figure}

\begin{table}[h]
\centering
\caption{Parameter maps reconstruction with HD-LM. Average MAE, SSIM and standard deviations, over voxels and subjects with respect to full-matching.}
\label{tab:mae_ssim_hdlm}
\begin{tabular}{lcccccc}
\hline
 & $T_1$ (ms) & $T_2$ (ms) & $\delta  f$ (Hz) & $B_1$ sensitivity ($10^{-3}$) & $CBV$ (\%) & R ($\mu m$) \\
\hline
MAE & $84 \pm 18$ & $12 \pm 3$ & $7 \pm 1$ & $45 \pm 7$ & $2.3 \pm 1$ & $1.9 \pm 0,2$ \\
SSIM & $0.8 \pm 0.02$ & $0.76 \pm 0.12$ & $0.4 \pm 0.05$ & $0.53 \pm 0.02$ & $0.38 \pm 0.07$ & $0.25 \pm 0.09$ \\
\hline
\end{tabular}
\end{table}

Finally, Figure~\ref{fig:recon_hdlm} shows  HD-LM reconstructions in comparison to full matching, for the different parameters and one slice of a subject. We observe some noise at the top of the $T_2$ map.

\begin{figure}[h]
    \centering
    \includegraphics[width=0.8\linewidth]{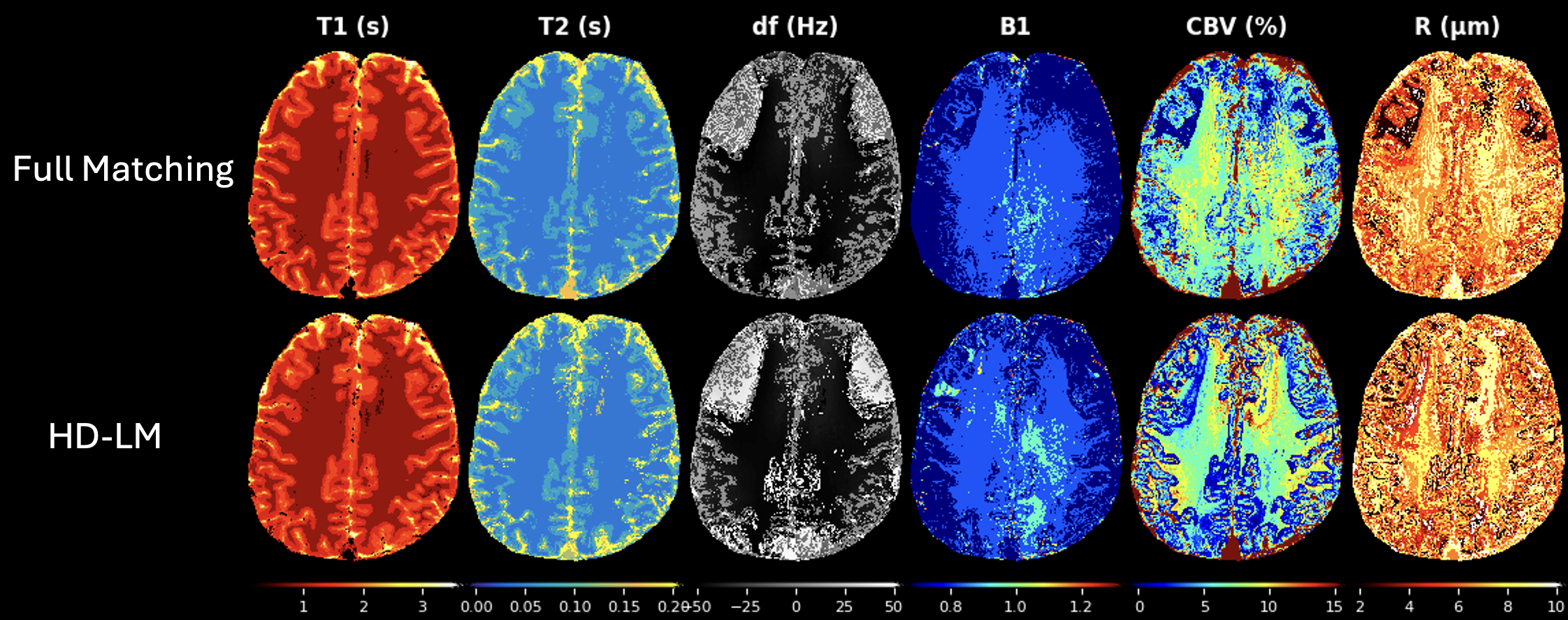}
    \caption{$T_1$, $T_2$, $\delta \! f$, $B_1$ sensitivity, $CBV$, $R$ maps (columns) with  \textit{full matching} (top) and HD-LM with $K=30$ (bottom)}
    \label{fig:recon_hdlm}
\end{figure}

\section{Divide \& Conquer high dimensional inverse regression of large data volumes}

As mentioned in the main paper conclusion, the presented data compression procedure could also be combined with a regression model in place of the matching step. Algorithm \ref{alg:DCR} below is a variant of Algorithm 1 in the paper, illustrating this possibility. Step 1 remains the same, Step 2 corresponds to the estimation of regression functions for each cluster and is new, Step 3 is the cluster-wise prediction of new signals using the $K$ estimated regression functions. More specifically, step 3.1 consists of projecting an observation $\widetilde{\yv}_j$ on each of the $K$ subspaces and of computing the probability of $\widetilde{\yv}_j$ to be assigned to each cluster. Step 3.2 then corresponds to a cluster-wise prediction when 
only the regression function for the most probable cluster is used, {\it i.e.} $\widetilde{\tv}_j = \widetilde{\tv}_j^k $ with $k =\arg\max\limits_l r_l(\widetilde{\yv}_j)$. Step 3.2 is written in a slightly more general manner using ${\cal L}$ to denote a function that would combine, in a way to be chosen, all predictions and all assignment probabilities. 

\begin{algorithm}[h]\scriptsize
    \caption{\small Divide \& Conquer high dimensional inverse regression of large data volumes}
    \label{alg:DCR}
    \begin{algorithmic}[1]
     \item[ ] {\bf Input} Data set of simulated (signal, parameters) pairs ${\cal D}_f = \{\yv_i, \tv_i\}_{i=1:N}$, $N>>1$, $\tv_i \in \Rset^L$, $\yv_i \in \Rset^M$, $M >> 1$.\\
    $\qquad$ Observed signals $\{\widetilde{\yv}_j\}_{j=1:\tilde{N}}$, $\widetilde{\yv}_j \in \Rset^M$. 
     \begin{tcolorbox}[colback=orange!10,colframe=orange!50!black, boxrule=0.01pt,boxsep=2pt, top=1pt, bottom=1pt, title style={halign=flush right,before upper=\strut}]
        \vspace{1pt}
        \STATE {\bf Reduced dimension representation of simulated data:} $\{\hat{\yv}_i, \tv_i, r.(\yv_i)\}_{i=1:N} $.
        
        \begin{tcolorbox}[colback=cyan!10,colframe=cyan!50!black, boxrule=0.01pt,boxsep=2pt, top=1pt, bottom=1pt, title style={halign=flush right,before upper=\strut}]
            \item[1.1] {\bf Online HD-MED inference from $\{\yv_i\}_{i=1:N}$:} $K$ clusters, $d_k<M$ for $k= 1:K$  $\quad \Longrightarrow$  cluster assignment probabilities and cluster-wise projections $(\rb,\Qb) =\{r_k(\cdot), Q_k(\cdot)\}_{k=1:K}$.
        \end{tcolorbox}

        \begin{tcolorbox}[colback=magenta!10,colframe=magenta!50!black, boxrule=0.01pt,boxsep=2pt, top=1pt, bottom=1pt, title style={halign=flush right,before upper=\strut}]
            \item[1.2] {\bf Cluster-wise signal reduction:}  $\qquad \{\yv_i\}_{i=1:N}, \rb, \Qb \quad \Longrightarrow \quad 
            \{ \widehat{\yv}_i=Q_k(\yv_i), i \in I_k\}\qquad $ with
            $\qquad I_k=\{i, \; s.t  \; k=\arg\max\limits_\ell r_\ell(\yv_i)\}$, for $k=1:K$.
        \end{tcolorbox}
    \end{tcolorbox}
     \begin{tcolorbox}[colback=cyan!10,colframe=cyan!50!black, boxrule=0.01pt,boxsep=2pt, top=1pt, bottom=1pt, title style={halign=flush right,before upper=\strut}]
            \STATE {\bf Cluster-wise regressions:} for $k=\!1\!:\!K$,  learn a  
           regression function $l_k:\Rset^{d_k} \rightarrow \Rset^L$  from  $\{ (\tv_i,\widehat{\yv}_i), \; i \in I_k\}$ .
    \end{tcolorbox}

        \begin{tcolorbox}[colback=green!10,colframe=cyan!50!black, boxrule=0.01pt,boxsep=2pt, top=1pt, bottom=1pt, title style={halign=flush right,before upper=\strut}]
            \STATE {\bf Cluster-wise predictions for observed signals:} 
             \begin{tcolorbox}[colback=cyan!10,colframe=cyan!50!black, boxrule=0.01pt,boxsep=2pt, top=1pt, bottom=1pt, title style={halign=flush right,before upper=\strut}]
            \item[3.1] {\bf Cluster-wise observed signal reduction:} Use learned $(\rb, \Qb)$ from \textit{step 1.1}  to obtain $\{Q_k(\widetilde{\yv}_j), r_k(\widetilde{\yv}_j),  j= 1:\tilde{N}, k= 1:K \}$
            \end{tcolorbox}
            \begin{tcolorbox}[colback=magenta!10,colframe=magenta!50!black, boxrule=0.01pt,boxsep=2pt, top=1pt, bottom=1pt, title style={halign=flush right,before upper=\strut}]
            \item[3.2]{\bf Cluster-wise prediction:} For $j= 1:\tilde{N}, k= 1:K$, compute $\widetilde{\tv}_j^k = l_k(Q_k(\widetilde{\yv}_j))$
            $\Longrightarrow \quad$ predict 
            $\widetilde{\tv}_j = {\cal L}(\widetilde{\tv}_{j}^1, r_{1}(\widetilde{\yv}_j), \ldots,\widetilde{\tv}_{j}^K, r_{K}(\widetilde{\yv}_j) )\quad $ {\it e.g.}  $\widetilde{\tv}_j = \widetilde{\tv}_j^k $ with $k =\arg\max\limits_l r_l(\widetilde{\yv}_j)$.
            \end{tcolorbox}
            
        \end{tcolorbox}
        \item[ ] {\bf Return}  Predicted parameters
        $ \{\widetilde{\tv}_j\}_{j=1:\tilde{N}}$, $\widetilde{\tv}_j \in \Rset^L$
    \end{algorithmic}
\end{algorithm}

\section{Application to MRF reconstruction}

We provide additional illustration by showing in Figure \ref{fig:addrecons} the parameter maps obtained with full matching and HD-STM, for all the slices of the subject presented in main Figure 6. Both methods, HD-STM and full matching provide visually similar reconstructions. 

\begin{figure}[h]
    \centering
    \begin{minipage}[t]{0.8\textwidth}
        \centering
        \includegraphics[width=\textwidth]{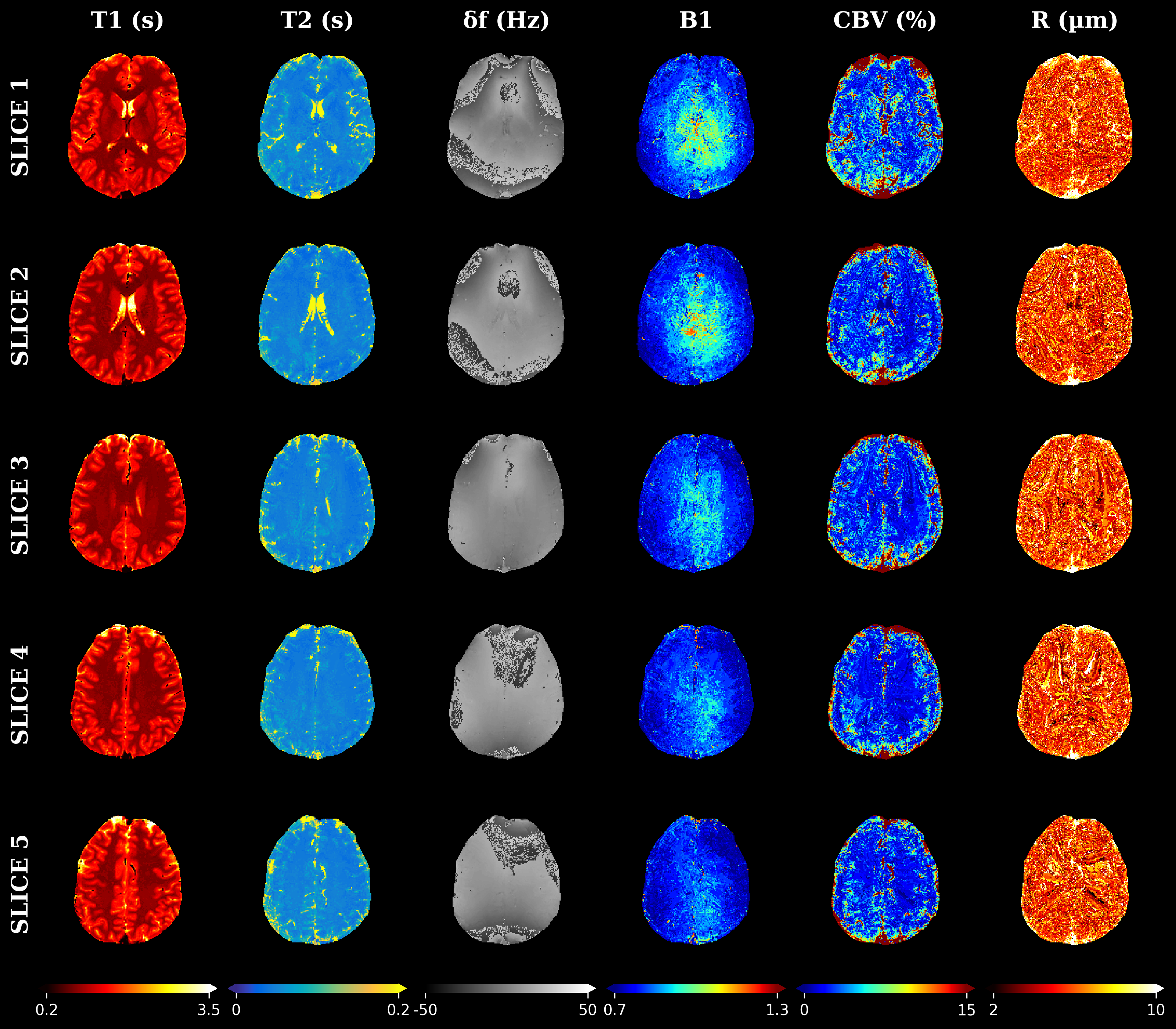}
    \end{minipage}
    \begin{minipage}[t]{0.8\textwidth}
        \centering
        \includegraphics[width=\textwidth]{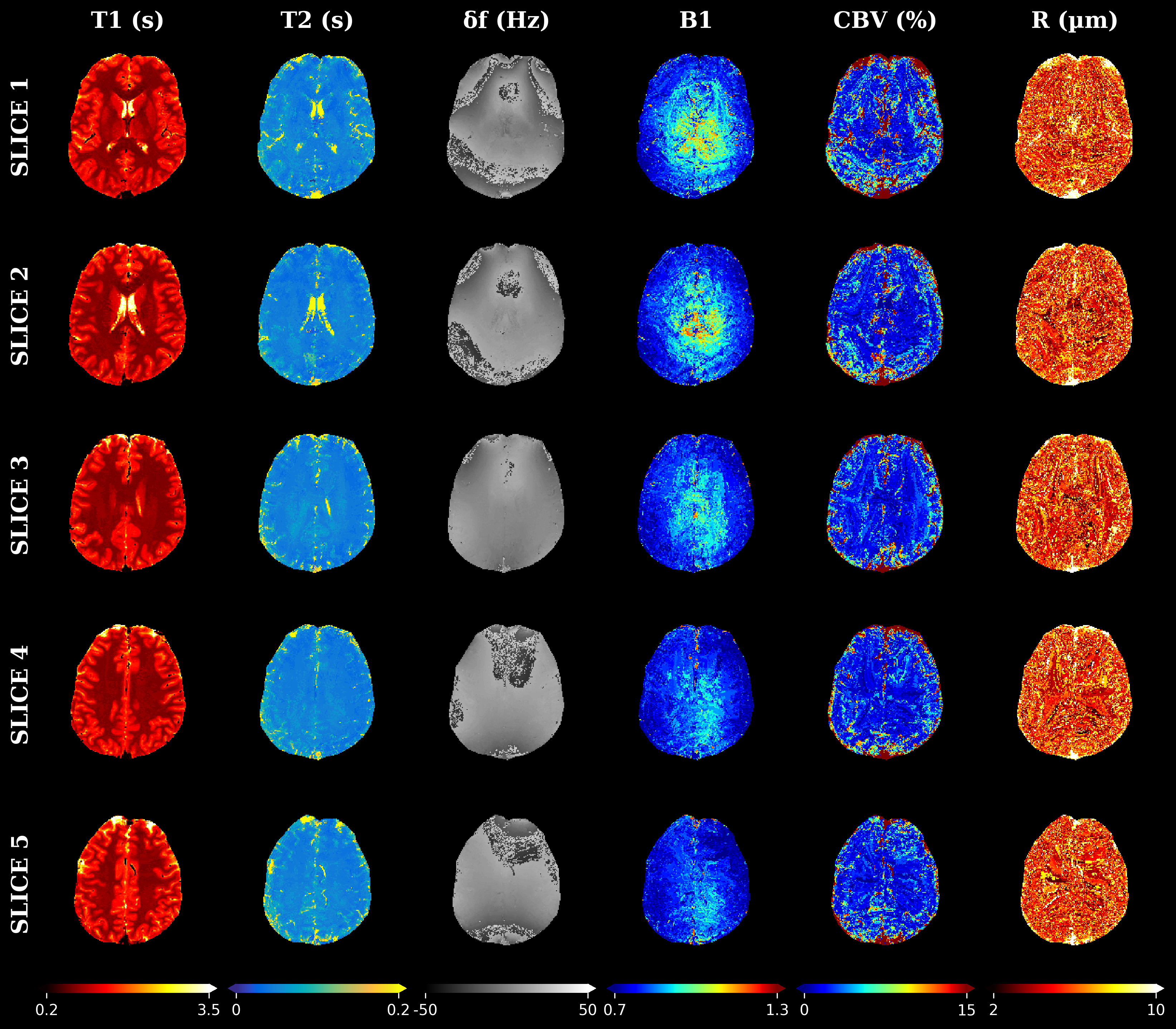}
    \end{minipage}
    \caption{$T_1$, $T_2$, $\delta \! f$, $B_1$ sensitivity, $CBV$, $R$ maps (columns) with different methods- \textit{full matching} (top), HD-STM with $K=25$ (bottom) for all the slices of the subject presented in Figure 6 of the main paper. \label{fig:addrecons}}

\end{figure}

\section{Generalized scale mixtures of Gaussian distributions}
In the main text, we focused on the case where the covariance matrix of each Gaussian component is scaled by a factor of $\frac{1}{w}$. However, as proposed in \cite{lee2014}, it is possible to generalize Gaussian scale mixtures by using a more flexible link function $k(\cdot)$, leading to the formulation:

\begin{equation}
    p(\yv) = \int \limits_{\Rset^+}  \mathcal{N}_M(\yv; \mub, k(w)\Sigmab) f(w) \diff w.
\end{equation}

This generalization encompasses a wider class of distributions, as summarized in Table 1 of \cite{lee2014}. A possible continuation of our work is to  derive properties from the main paper for this more general case, while still applying the proposed parcimonious parameterization of high dimensional scale matrices. It is noteworthy that we retain the ability to define a latent variable model for this extended class of distributions. Using the same notation, the following proposition holds for this new class, denoted by $\mathcal{HE}$ as in the main paper:

\begin{proposition}
Let $d \leq M-1$, $\Yv \in \Rset^M$,  $\Xv \in \Rset^d$, $\Eb \in \Rset^M$, $W \in \Rset^+$ be random variables,  $\Vb \in \Rset^{M \times d}$ a matrix of linearly independent columns, $\mub \in \Rset^M$ a vector, a link function $k: \Rset \xrightarrow{}\Rset$ and $f_\thetab$ the pdf of a positive univariate random variable defined by some parameter $\thetab$. Assume that
\begin{align*}
\Yv &= \Vb \Xv + \mub + \Eb \\
(\Xv | W=w) &\sim {\cal N}(\mathbf{0}_d, k(w) \Ib_d) \\
(\Eb | W=w) &\sim {\cal N}(\mathbf{0}_M, b \; k(w) \Ib_M) \\
W &\sim f_\thetab
\end{align*}
then, $$\Yv \sim \mathcal{HE}_{Md} \left(\mub, \widetilde{\Db}^*, \ab, b, \thetab \right),$$ 
with $b \Ib_M + \Vb \Vb^T = \Db \Ab \Db^T$ and $\Ab = \operatorname{diag}(a_1, \dots, a_{d}, b, \dots, b)$.

Additionally, denoting by $\Ub = b \Ib_d + \Vb^T \Vb $, we have, 
\begin{align}
(\Xv | \Yv=\yv, W=w) &\sim {\cal N}( \Ub^{-1}\Vb^T(\yv - \mub), k(w) b \Ub^{-1}) \label{eq:condX} \; .
\end{align}
It follows that $\mathbb{E}[\Yv | \Xv=\xv ] =  \Vb\xv +\mub$ and 
$\mathbb{E}[\Xv | \Yv=\yv ] =  \Ub^{-1}\Vb^T(\yv - \mub)\; .$
\end{proposition}

\begin{proof}
    The proof follows the same structure as Proposition 5.1 in the main paper, with the scaling factor $\frac{1}{w}$ replaced by the link function $k(w)$.
\end{proof}

Regarding the OEM part, we can similarly derive propositions for this new class of distributions, assuming that $\Wv$ follows a distribution in the exponential family. For all link functions $k(w)$, most calculations involve replacing the scaling factor $\frac{1}{w}$ with $k(w)$. This opens the possibility of studying high-dimensional finite mixtures of distributions such as the slash, variance gamma, and Laplace distributions.

\bibliographystyle{imsart-nameyear}
\bibliography{bibliography}